\documentclass[12pt]{article}

\usepackage{amsmath,amssymb,amsthm,cite,slashed,mathtools}
\usepackage[margin=3cm]{geometry}

\usepackage{fontaxes}
\usepackage[T1]{fontenc}
\usepackage[osf]{newtxtext} 
\usepackage[bigdelims,nosymbolsc,cmintegrals,vvarbb]{newtxmath}
\usepackage[cal=euler,frak=euler]{mathalfa} 
\usepackage{bm} 
\usepackage[colorlinks]{hyperref}
\hypersetup{linkcolor = blue, citecolor = green!75!black}
\usepackage[normalem]{ulem}

\usepackage{tikz}
\usetikzlibrary{arrows.meta, calc, positioning, 
decorations.pathmorphing,decorations.markings}
   
\tikzset{photon/.style={decorate, decoration={snake}},
   wplus/.style={postaction={decorate},
   decoration={markings, mark=at position 0.55 with {\arrow{>}}}}}

\title{Diphoton decay of the higgs \\ from the Epstein--Glaser viewpoint}

\author{Pawe{\l} Duch$^{1,2}$,
Michael D\"utsch$^3$ and Jos\'e M. Gracia-Bond\'ia$^{4,5}$%
\footnote{Email: jmgb@unizar.es}
\\[6pt]
{\footnotesize $^1$ Institut f\"ur Theoretische Physik,
Universit\"at Leipzig, 04103 Leipzig, Germany}
\\[3pt]
{\footnotesize $^2$ Max-Planck Institute for Mathematics in the
Sciences, 04103 Leipzig, Germany}
\\[3pt]
{\footnotesize $^3$ Institut f\"ur Theoretische Physik, Universit\"at
G\"ottingen, G\"ottingen D-37077, Germany}
\\[3pt]
{\footnotesize $^4$ CAPA and Departamento de F\'isica Te\'orica,
Universidad de Zaragoza, Zaragoza 50009, Spain}
\\[3pt]
{\footnotesize $^5$ Laboratorio de F\'isica Te\'orica y Computacional,
Universidad de Costa Rica, San Pedro 11501, Costa Rica}
}

\date{\today}

\newcommand{\version}{Agla\'e: \today}
\DeclareMathOperator{\sd}{sd}       
\DeclareMathOperator{\sgn}{sgn}     
\DeclareMathOperator{\supp}{supp}   
\DeclareMathOperator{\T}{T}         

\newcommand{\al}{\alpha}            
\newcommand{\bt}{\beta}             
\newcommand{\Dl}{\Delta}            
\newcommand{\dl}{\delta}            
\newcommand{\Ga}{\Gamma}            
\newcommand{\ga}{\gamma}            
\newcommand{\La}{\Lambda}           
\newcommand{\la}{\lambda}           
\newcommand{\Om}{\Omega}            
\newcommand{\om}{\omega}            
\renewcommand{\th}{\theta}          
\newcommand{\vf}{\varphi}           

\newcommand{\bC}{\mathbb{C}}        
\newcommand{\bM}{{\mathbb{M}}}      
\newcommand{\bN}{\mathbb{N}}        
\newcommand{\bR}{{\mathbb{R}}}      
\newcommand{\bS}{\mathbb{S}}        

\newcommand{\sA}{\mathcal{A}}       
\newcommand{\sK}{\mathcal{K}}       
\newcommand{\sM}{\mathcal{M}}       
\newcommand{\sR}{\mathcal{R}}       
\newcommand{\sS}{\mathcal{S}}       
\newcommand{\sU}{\mathcal{U}}       
\newcommand{\sV}{\mathcal{V}} 

\newcommand{\const}{\mathrm{const}}   
\newcommand{\gi}{\mathrm{gi}}         
\newcommand{\naive}{\mathrm{naive}}   

\newcommand{\del}{\partial}         
\newcommand{\downto}{\downarrow}    
\newcommand{\less}{\setminus}       
\newcommand{\longto}{\longrightarrow} 
\newcommand{\otto}{\leftrightarrow} 
\newcommand{\ovl}{\overline}        
\newcommand{\rtri}{\blacktriangleright} 
\newcommand{\vecvec}{\overleftrightarrow} 
\newcommand{\wt}{\widetilde}        
\newcommand{\x}{\times}             
\newcommand{\7}{\dagger}            
\newcommand{\8}{\bullet}            
\renewcommand{\.}{\cdot}            

\newcommand{\half}{{\mathchoice{\thalf}{\thalf}{\shalf}{\shalf}}}
\newcommand{\ihalf}{\tfrac{i}{2}}   
\newcommand{\quarter}{\tfrac{1}{4}} 
\newcommand{\shalf}{{\scriptstyle\frac{1}{2}}} 
\newcommand{\thalf}{\tfrac{1}{2}}   

\bmdefine{\kk}{k}                   
\bmdefine{\PP}{P}                   
\bmdefine{\pp}{p}                   
\bmdefine{\qq}{q}                   
\bmdefine{\yy}{y}                   
\bmdefine{\zero}{0}                 

\newcommand{\piso}[1]{\lfloor#1\rfloor} 
\newcommand{\set}[1]{\{\,#1\,\}}    
\newcommand{\vev}[1]{\langle\!\langle#1\rangle\!\rangle} 
\newcommand{\word}[1]{\quad\text{#1}\quad} 

\newcommand{\braket}[2]{\langle#1\mathbin|#2\rangle} 

\newcommand{\marker}{\vspace{6pt}\noindent{$\rtri$}\enspace} 

\theoremstyle{plain}
\newtheorem{thm}{Theorem}           
\newtheorem{lema}[thm]{Lemma}       
\newtheorem{corl}[thm]{Corollary}   

\theoremstyle{definition}

\theoremstyle{remark}
\newtheorem{remk}{Remark}           

\numberwithin{equation}{section}
\numberwithin{figure}{section}

\makeatletter
\renewcommand{\section}{\@startsection{section}{1}{\z@}%
                        {-3.5ex \@plus -1ex \@minus -.2ex}%
                        {2.3ex \@plus.2ex}%
                        {\normalfont\large\bfseries}}
\renewcommand{\subsection}{\@startsection{subsection}{2}{\z@}%
                        {-3.25ex \@plus -1ex \@minus -.2ex}%
                        {1.5ex \@plus .2ex}%
                        {\normalfont\normalsize\bfseries}}
\renewcommand{\subsubsection}{\@startsection{subsubsection}{3}{\z@}%
                        {-3.25ex \@plus -1ex \@minus -.2ex}%
                        {1.5ex \@plus .2ex}%
                        {\normalfont\normalsize\itshape}}
\renewcommand{\@dotsep}{200} 
\makeatother

\begin{document}

\maketitle

\begin{flushright}
\textit{To the memory of G\"unter Scharf and Raymond Stora}
\end{flushright}

\medskip

\begin{abstract}
We revisit a nearly ten-year old controversy on the diphoton decay of
the Higgs particle. To a large extent, the controversy turned around
the respective merits of the regularization techniques employed. The
novel aspect of our approach is that \textit{no} regularization
techniques are brought to bear: we work within the
Bogoliubov--Epstein--Glaser scheme of renormalization by extension of
distributions. Solving the problem actually required an expansion of
this method's toolkit, furnished in the paper.
\end{abstract}

\begin{flushright}
\textit{Die Eule der Minerva beginnt erst mit der einbrechenden
D\"ammerung ihren Flug}

\medskip

-- Georg Wilhelm Friedrich Hegel
\end{flushright}

\medskip
%

\section{Introduction: the controversy}
\label{sec:intro}

Due to its cleanness, it is hard to overstate the experimental
importance of the decay of the Higgs particle into two photons. It
goes mainly via virtual $W$-bosons, the heavier charged particles of
flavourdynamics. The amplitude of this contribution was calculated to
the first non-vanishing order (one-loop, cubic in the couplings) long
ago in the light-higgs limit~\cite{Eliseo} -- and then ``exactly''
in~\cite{RussianCharge2}. The accepted result was confirmed many times
-- see~\cite{elBelen} for a particularly clever calculation. It does
\textit{not} vanish in the heavy-higgs limit -- which seems to fly in
the face of the ``decoupling theorem'' (DT) in~\cite{Loscarrozas}, as
often understood.

Much more recently, those calculations were questioned in
\cite{GastmansWuWu1, GastmansWuWu2}. The ensuing debate highlights the
\textit{theoretical} relevance of this decay. The authors of these
papers made the point that, since the higgs cannot couple directly to
the photons, the one-loop contribution must be finite: there are no
couplings requiring ``renormalization''. The roundabout procedures
through ``renormalizable gauges'', they concluded, were unnecessary.
Eschewing dimensional regularization, they recomputed the amplitude in
the unitary gauge of electroweak (EW) theory. They did obtain a result
differing from the standard one by an additive constant, which shows
up for instance in the heavy-higgs limit -- whereby their result is
equal to zero.

There was no shortage of rejoinders~\cite{TresChinosCutoff,
MartianCharge, RussianCharge3, TresChinosAgudos, Jaguarlehner,
PichininiPP, DedesS, Weinzierl14} to \cite{GastmansWuWu1,
GastmansWuWu2}. The authors of~\cite{RussianCharge3} are the ones of
the original calculation~\cite{RussianCharge2}. Those papers made
several points, some rather implausibly arguing that at a given point
in the calculation in \cite{GastmansWuWu2} electromagnetic gauge
invariance is lost, and criticizing the interpretation of the DT made
in~\cite{GastmansWuWu1, GastmansWuWu2}. There was in some of the the
rejoinders an explanatory reliance on the heuristics of the
Brout--Englert--Higgs mechanism, throwing back the so-called
``equivalence theorem'' (GBET).

The criticisms received a rejoinder in turn in~\cite{WusStrikeBack}.
This later paper argues by the example that two computations of the
same process in different gauges ($R_\xi$ versus unitary gauge) may
yield different results. This goes against the grain, although of
course no theorem contradicts such an assertion. Meanwhile, a
dispersion relation calculation carried out in~\cite{ChristovaI}
appeared to support the contentions of~\cite{GastmansWuWu1,
GastmansWuWu2}, and got in turn a -- quite thoughtful -- rejoinder
in~\cite{MelnikovVainshtein}. More recent papers dealing with the same
or related issues are \cite{GegeliaM18,BoraChristovaEberl}.

By and large, the majority's opinion and the experimental results
\cite{Jacobs} support the first tally. On the other hand, from the
theoretical point of view the situation is still obscure: it had to be
so, since both parties draw strength from different casuistics of the
calculations in perturbative quantum field theory.

\smallskip

The debate about the uses and abuses of the unitary gauge and the role
of the decoupling and equivalence ``theorems'' is to be saluted as
salutary. And it is safe to admit that up to now we lack a full
conceptual understanding of the problem. The cleanest way to address
this lack is surely to renounce \textit{all} the heuristics of
mathematically ill-defined quantities, in favour of a method in which
there can be no argument on the meaning of infinite terms. Such is the
truly (perturbatively) stringent scheme by Bogoliubov, Epstein and
Glaser (BEG) of ``renormalization'' without regularization, by
extension of distributions.

In the BEG construction, governed by causality, there is no such thing
as a ``divergent diagram'': one never encounters infinities. There
may, however, remain in the extension procedures some additive
\textit{ambiguity}, that can be restricted (but not always completely
removed) by physical principles. This is rather to be regarded as a
strength of the BEG paradigm, because those ambiguities express
precisely how, and to what extent, the theory is determined by the
fundamental principles of perturbative QFT.

A particular advantage of the inductive BEG construction
\cite{EpsteinGlaser73} of the (functional) $\bS$-matrix is that in
principle one is allowed to stay on configuration space, which makes
more transparent the physics under examination. For examples of
calculations within the BEG scheme explicitly carried out in
configuration space, see \cite{Elara} or \cite[Sect.~3.5]{Duetsch19}.
It is only for computational convenience that we switch at some moment
to momentum space.

Since we do not deal in infinities, we refer as \textit{normalization}
to the processes taking the place of regularization and
renormalization in the BEG framework. For its relative paucity of
diagrams, in our context the underlying argument is made clearer by
working mostly in the unitary gauge -- whereupon only the
physical particles' data are brought to bear.%
\footnote{The paper~\cite{IrgesKoutroulis17} dwells usefully on the
subject of the $R_\xi$-versus-unitary gauges, leaning to demonstrate
the validity of the latter at the quantum level.}

To summarize, so far: we were motivated to tackle this subject by
wondering why most knowledgeable people, borrowing different (but all
apparently sound) methods to work on such a basic process, were
divided on the outcome. It all turns around a subtlety uncovered by
use of the BEG normalization. That condenses the purpose of the
present paper.

\subsection{Main results and plan of the article}
\label{ssc:facile-est}

In App.~\ref{app:spin-one-basics} we introduce our conventions and
notations, recalling a few well-known formulae of QFT needed in the
body of the paper, in particular the propagators for the EW theory in
the unitary gauge. Let $m_h$ denote the mass of the higgs~$h$. The
amplitude coming from the one-loop calculations may be quoted as
\cite{HungryHunters, BardinP99, GranLev}:
$$
\sA = \frac{g\al}{2\pi M} F_1(\rho) P_{\mu\nu},
$$
with $\al$ the fine structure constant, $g$ the EW coupling constant,
$M$ the mass of the intermediate $W$-boson and $\rho := m^2_h/4M^2$.
The polarization factor $P_{\mu\nu}$, reflecting electromagnetic gauge
invariance (EGI) of~$\sA$,%
\footnote{That is, transversality of the outgoing photons.}
is written in this paper as
\begin{equation}
P_{\mu\nu} := (k_1 k_2) g_{\mu\nu} - k_{1\nu} k_{2\mu}; \quad
(P_{\8\nu} k_1) = (P_{\mu\8} k_2) = 0,
\label{eq:in-tractu-temporis} 
\end{equation}
with $k_1$, $k_2$ the outgoing photons' momenta. Finally, for the 
dimensionless factor:
\begin{equation}
F_1(\rho) := 2 + \frac{3}{\rho} 
+ \frac{3}{\rho} \biggl( 2 - \frac{1}{\rho} \biggr) f(\rho).
\label{eq:bone-of-contention} 
\end{equation}
Now that we are at that, we quote as well the comparable result for a
charged scalar particle of mass~$M$ at the place of the $W$-boson:
\begin{equation}
F_0(\rho) = \frac{1}{\rho} \biggl( 1 - \frac{f(\rho)}{\rho} \biggr);
\word{so that} F_1(\rho) = 3F_0(\rho) + \frac{6f(\rho)}{\rho} + 2.
\label{eq:magister-dixit} 
\end{equation}
For the benefit of the reader coming to the subject of this paper for
the first time, App.~\ref{app:curioser} introduces the distribution
$f(\rho)$ appearing in both $F_1$~\eqref{eq:bone-of-contention} and
$F_0$~\eqref{eq:magister-dixit} -- as well as in the amplitude of
diphoton decay of~$h$ via virtual fermions.

The bone of contention is that the first summand~$2$ in
\eqref{eq:bone-of-contention} should not be there, according to
\cite{GastmansWuWu1, GastmansWuWu2, ChristovaI}. Relations
\eqref{eq:f1} and \eqref{eq:arcsin-2} tell us that, as $\rho\downto0$:
$$
F_1 = 2 + \frac{3}{\rho} 
+ \biggl( \frac{6}{\rho} - \frac{3}{\rho^2} \biggr)
\biggl( \rho + \frac{\rho^2}{3} + \frac{8\rho^3}{45} +\cdots \biggr)
= 7 + \frac{22}{15}\,\rho + O(\rho^2);
$$
so $F_1(0) = 7$ and $F_1(\infty) = 2$ from
\eqref{eq:bone-of-contention}. Precisely the former figure is what was
calculated in the paper~\cite{Eliseo}. The result argued by the
``heretics'' in the controversy is $F_1 - 2$, so their respective
assertions are instead $F_1(0) = 5$ and $F_1(\infty) = 0$. Also, from
\eqref{eq:magister-dixit}: $F_0(0) = -1/3$ and $F_0(\infty) = 0$.

\smallskip

Appendices A and B of this paper deal with conventions and
mathematical prerequisites. The basics of the BEG scheme are recalled
in Appendix~\ref{sec:Streu}. Understanding of the BEG method is
indispensable in what follows, and even readers familiar with it are
advised not to miss our review. The relation between the normalization
problem by extension of distributions (or by ``distribution
splitting'') and \textit{dispersion integrals} is treated in its
subsection~\ref{sec:ipso-facto}. New results in this respect are
required, announced in the short Section \ref{sec:quam-scriptum} and
proved in subsections \ref{ssc:do-ut-des} and
\ref{ssc:hasta-ahi-podiamos-llegar} of this paper. So for
\textit{aficionados} of BEG normalization there is novelty here --
whose interest goes beyond the particular problem that motivated~it.

Sections~\ref{sec:argumentum} and ~\ref{sec:soberbia-pagana}
constitute the heart of the paper. The scalar model leading to $F_0$
is worked out in Section~\ref{sec:argumentum}. One is able to perform
the ``adiabatic limit'' of Epstein and Glaser at an intermediate step,
which simplifies computations -- this is rigorously justified. This
``toy model'' allows the reader to familiarize with the BEG
construction of time-ordered products in a relatively simple case. For
it, the ambiguity in the Epstein--Glaser result can be disposed of,
and the unique outcome happens to coincide with the result of a
``naive'' on-shell calculation, of the kind performed
in~\cite{ChristovaI}.

Finally, in Section~\ref{sec:soberbia-pagana}, we compute the EW
amplitude, working first in the unitary gauge. We start in earnest by
illustrating in this relevant instance the machinery of the BEG
formalism in constructing time-ordered products, at the lowest
non-trivial order: from cubic interaction vertices, identified to
time-ordered products at first order in the couplings, we
\textit{derive} the quartic, second-order $AAWW^\7$-vertex.

It is time to aver why the ``no-renormalization'' argument in
\cite{GastmansWuWu2} is not watertight. A direct $h\ga\ga$~coupling in
flavourdynamics is forbidden also because of EGI. Thus to obtain the
general amplitude, which lives off-shell, one must add to the naive
calculations a polynomial in the external momenta, of degree given by
the singular order of that amplitude. Computing the $1$-loop
contribution in the unitary gauge by the Epstein-Glaser method, we
ratify this fact. To find the coefficients of that polynomial, beyond
EGI here we call upon gauge-fixing independence of the on-shell
amplitude. This locks in the indetermination; and in the end we do
obtain~$F_1(\rho)$. Within the unitary gauge, a different argument to
the same purpose is discussed at the end of this
Section~\ref{sec:soberbia-pagana}. Section~\ref{sec:die-Eule} is the
conclusion.

\section{The obstruction to distribution splitting for null momenta}
\label{sec:quam-scriptum}

Formula \eqref{eq:pro-reo} in App.~\ref{sec:Streu} is our main
workhorse: in momentum space the Epstein--Glaser distribution
splitting amounts to a dispersion integral. But it pertains to remark
that, by construction, prescriptions \eqref{eq:hoist-with-retard}
and~\eqref{eq:pro-reo} are in principle valid \textit{only for
timelike}~$k$. Thus, in order to solve the problem in this paper, one
has to run an extra~mile. The explicit splitting procedure introduced
here exhibits relevant novel features: we have to compute the central
solution $a^c(k_1,k_2)$ for \textit{null momenta}. Hence, one cannot
immediately use the dispersion integrals \eqref{eq:hoist-with-retard}
or~\eqref{eq:pro-reo}. On trying to work instead with the convolution
integral \eqref{eq:ac}, there appears the problem that, in spite of
$k_j^2 = 0$, it generally holds that $(k_j - v_j)^2 \neq 0$ because
$v_j \in V_+$; it \textit{does not suffice} to know the causal
distribution $d(k_1,k_2)$ only for $k_1^2 = 0 = k_2^2$.

The next section solves this problem for models such that $0 < (k_1 +
k_2)^2 < 4M^2$ and $k_1^0 k_2^0 > 0$. The proof's strategy is as
follows: starting from the dispersion
integral~\eqref{eq:hoist-with-retard} for $k_1^2 > 0$, $k_2^2 > 0$ and
$k_1^0 k_2^0 > 0$, we intend to show that $d(k_1,k_2)$ is regular
enough that this integral commutes with the limit $(k_1^2 \downto 0
\wedge k_2^2 \downto 0)$. Therefore the dispersion
integrals~\eqref{eq:hoist-with-retard} and~\eqref{eq:pro-reo} keep
their usefulness for $k_1^2 = 0 = k_2^2$: indeed, for computing
$a^c(k_1,k_2)|_{k_1^2=0=k_2^2}$ it suffices to know $d(k_1,k_2)$ only
for $k_1^2 = k_2^2 = 0$, because $k_1^2 = k_2^2 = 0$ implies $(tk_1)^2
= (tk_2)^2 = 0$ for all~$t$.

Crucially, in the resulting dispersion integrals
\eqref{eq:hoist-with-retard} and \eqref{eq:pro-reo} for $k_1^2 = 0 =
k_2^2$, the parameter $\om$ is the singular order of the
\textit{off-shell} $d(k_1,k_2)$. As a consequence, the general
solution (prior to imposition of other invariance rules) of the
distribution splitting is obtained by adding to
$a^c(k_1,k_2)|_{k_1^2=0=k_2^2}$ a polynomial in $k_1,k_2$, in
principle arbitrary, whose degree is given by the singular order of
the off-shell amplitude $d(k_1,k_2)$. Now, it frequently happens that
the singular order of $d(k_1,k_2)|_{k_1^2=0=k_2^2}$ has a
\textit{smaller} value. Consequently, it may happen that the required
dispersion integral appears to be ``oversubtracted'' -- i.e., it would
be convergent also for a smaller value of~$\om$. Examples for this are
the ``toy model'' in the next section and the EW diphoton decay of the
higgs in the unitary gauge (subsections
\ref{ssc:hasta-ahi-podiamos-llegar} and \ref{ssc:quod-feceris},
respectively).

These issues were realized by Raymond Stora, who, referring to the
very subject process of this paper, pointed out to one of us that the
good behaviour of the absorptive part of the form factor involving
Compton scattering of the $W$-bosons should not make one forget that
BEG-generated dispersion integrals, just as perturbative
renormalization theory in general, applies off-shell.%
\footnote{Private communication, early 2013.}

\section{Higgs to diphoton decay via a charged scalar field}
\label{sec:argumentum}

The scalar electrodynamics computation leading to~$F_0$ works like a
kind of toy model, allowing the reader to familiarize with our methods
in a less complicated, although non-trivial case. We develop it in the
present section. Notice the following: in the Epstein--Glaser scheme
the ``seagull'' $e^2 AA\vf\vf^\7$-vertex \textit{is derived} by
implementing EGI within the construction rules of the method -- as any
other part of $T_2$~\cite{DKS93}. We give full details on how this
comes about for the quartic vertex in the EW theory in subsection
\ref{ssc:per-angusta}. The game here would be similar, only simpler.
The reader is advised to keep in mind the methods and standard
notations recalled in subsection~\ref{ssc:lost-in-translation}.

\subsection{A causal distribution on-shell}
\label{ssc:d-scalar}

The starting point is given by the lower order time-ordered products
(TOPs):
\begin{align*}
&T_1(x_3) = g M\,h(x_3)\, \vf(x_3)\, \vf^\7(x_3); \quad
T_1(x_j) = -ie A^\la(x_j)\, \vf^\7(x_j) \vecvec{\del_\la} \vf(x_j),
\;\; j = 1,2;
\\
& T_2(x_1,x_2)
\\
&\quad = -e^2 A^\mu(x_1) A^\nu(x_2) \bigl[ \vf^\7(x_1)
\,\del_\mu \Delta^F(x_1 - x_2) \,\del_\nu\vf(x_2) -
\del_\mu\vf^\7(x_1) \Delta^F(x_1 - x_2) \,\del_\nu\vf(x_2)
\\
&\quad
+ \vf^\7(x_1)\, \bigl( \del_\nu\del_\mu \Delta^F(x_1 - x_2)
+ i g_{\mu\nu} \,\dl(x_1 - x_2) \bigr) \vf(x_2)
- \del_\mu\vf^\7(x_1) \,\del_\nu \Delta^F(x_1 - x_2)\, \vf(x_2) \bigr]
\\
&\quad + (x_1 \otto x_2) + T_1(x_1)\,T_1(x_2) + [\text{irrelevant loop
diagram terms}],
\end{align*}
where $\Dl^F$ denotes the Feynman propagator~\eqref{eq:post-factum}.

From our formulas~\eqref{eq:quod-non-speratur1} and
\eqref{eq:quod-non-speratur2}:%
\footnote{The $D_n$ are always linear combinations of commutators.}
\begin{equation}
D_3(x_1,x_2,x_3) = -[\ovl{T}_1(x_1), T_2(x_2, x_3)]
- [\ovl{T}_1(x_2), T_2(x_1, x_3)] + [\ovl{T}_2(x_1, x_2), T_1(x_3)].
\label{eq:para-llorar} 
\end{equation}

Because the photons emitted at $x_1,x_2$ are on-shell, only the third
commutator is relevant here -- in the language of Cutkosky rules, one
needs only the triangle cut separating the higgs vertex from the
propagator connecting the photons. We give the explanation further on.
From the general formula for the antichronological product
\eqref{eq:full-house}, we particularly know that
\begin{equation}
\ovl{T}_1(x_1) = T_1(x_1); \quad 
\ovl{T}_2(x_1, x_2) 
= - T_2(x_1, x_2) + T_1(x_1) T_1(x_2) + T_1(x_2) T_1(x_1).
\label{eq:ovlT} 
\end{equation}
For the same reasons just argued, only the connected tree diagram part
of the $T_2(x_1, x_2)$ summand in $\ovl{T}_2(x_1, x_2)$ contributes.

A most convenient parallel for the coming calculation is the treatment
of the vertex function in QED in the first edition of the finite QED
book by Scharf~\cite[Sect.~3.8]{Scharf89}. Going to the contractions,
bringing in the vertices and the propagators
\eqref{eq:accidit-in-puncto}, \eqref{eq:nullius-in-verba}, apart from
a factor $4ge^2M$ we obtain:
\begin{align*}
& A_\mu(x_1) A_\nu(x_2) h(x_3) \bigl[
\Dl^-(1)\,\del^\mu \Delta^F(1 - 2)\,\del^\nu\Dl^-(2)
- \del^\mu\Dl^-(1)\,\del^\nu \Delta^F(1 - 2)\,\Dl^-(2)
\\
&\quad - \del^\mu\Dl^-(1)\,\Delta^F(1 - 2)\,\del^\nu\Dl^-(2)
+ \Dl^-(1)\bigl( \del^\mu\del^\nu \Delta^F(1 - 2)
+ i g^{\mu\nu}\,\dl(1 - 2) \bigr) \Dl^-(2)
\\
&\quad -[\text{the same four terms with $\Delta^-$ replaced by
$\Delta^+$}] +\cdots \bigr]
\\
&=: A_\mu(x_1) A_\nu(x_2) h(x_3)\, d^{\mu\nu}(1,2),
\end{align*}
where $1 \equiv y_1 := x_1 - x_3$, $2 \equiv y_2 := x_2 - x_3$. Here
and further down, the dots stand for the terms coming from the other
two cuts and further terms not contributing to the on-shell amplitude.
Note the advertised additional $+i g^{\mu\nu}\,\dl$ to
$\del^\mu\del^\nu \Delta^F$, corresponding to the ``closed seagull''
or fish-like diagram contribution to the $h \to 2\ga$ decay in this
model.

\smallskip

We now proceed to momentum space, where computations are carried out
more simply. For Fourier transformations, consult the convention
\eqref{eq:def_Fourier}. In this section and the next, in keeping with
physicists' notation, we indicate the transforms by just exhibiting
the variables, namely: $d^\mu(k_1,k_2) \equiv \hat d^\mu(k_1,k_2)$. We
obtain
\begin{align}
d^{\mu\nu}(k_1,k_2) 
= & \frac{1}{(2\pi)^2} \bigl[
4(I_+^{\mu\nu} - I_-^{\mu\nu}) + 2 k_2^\nu(I_+^\mu - I_-^\mu)
- 2 k_1^\mu(I_+^\nu - I_-^\nu) - k_1^\mu k_2^\nu(I_+ - I_-)
\notag
\\
&\quad
- \frac{i}{(2\pi)^2}\, g^{\mu\nu}\, (J_+ - J_-) \bigr] +\cdots
\label{eq:suerte-o-verdad} 
\end{align}
with the integrals
\begin{align}
I_\pm^{\{\cdot|\mu|\mu\nu\}}(k_1,k_2)
&:= \int d^4k\, \{1|k^\mu|k^\mu k^\nu\} \,\Dl^\pm(k_1 - k) \Delta^F(k)
\,\Dl^\pm(k + k_2),
\nonumber \\
J_\pm(k_1,k_2) &:= \int d^4k\, \Dl^\pm(k_1 - k)\,\Dl^\pm(k + k_2),
\label{eq:def-J} 
\end{align}
where the $J_\pm$-term is the contribution of the fish-like diagram.
Keep in mind that the terms belonging to~$A'_3:=A_3-T_3$ are those
coming from the integrals $I^{\cdot|\mu|\mu\nu}_-$ and $J_-$, whereas
the contribution of~$R'_3:=R_3-T_3$ is given by the integrals
$I^{\cdot|\mu|\mu\nu}_+$ and~$J_+$. 

For our purposes one may perform the adiabatic limit already at this
stage. Since all internal lines of the diagrams correspond to massive
fields, this limit can be done here in the naive way by just setting
the switching function $g(x)$ in~\eqref{EG-summacumlaude} to~$1$:
\begin{align}
& \int dx_1\,dx_2\,dx_3\, A^\mu(x_1) A^\nu(x_2) h(x_3)\,
d^{\mu\nu}(x_1 - x_3, x_2 - x_3)
\notag
\\
&\enspace = (2\pi)^2 \int dk_1\,dk_2\, h(k_1 + k_2) A^\mu(-k_1)
A^\nu(-k_2) \,d^{\mu\nu}(k_1,k_2).
\label{eq:dAAh} 
\end{align}
In this limit the momenta $k_1$ and $k_2$ become the momenta of the
external photons: $k_1^2 = k_2^2 = 0$.

From now on, we compute $d^{\mu\nu}(k_1,k_2)|_{k_1^2=0=k_2^2}$. Were
we to have included the other cuts in \eqref{eq:para-llorar} or
$T_1T_1T_1$-terms, there would appear $\Dl^\pm$-type propagators at
the place of the Feynman propagators above. The former are $\sim
\dl(k^2 - M^2)$, with $k$ denoting the internal momentum variable in
the loop: so to speak, in contrast with the Feynman
propagators, the $\Dl^\pm$ are ``always on-shell'', even within loops.%
\footnote{This point is made in \cite[Sect.~6.4]{Diag}.}
Thus no further internal momenta can be on-shell: assuming $k^2 = M^2$
one obtains $(k_1 - k)^2 = M^2 - 2(k_1k) \neq M^2$;
similarly for $(k + k_2)$.%
\footnote{Compare the discussion after \cite[Eq.~(3.8.24)]{Scharf89}.}

\paragraph{Scalar integrals $I_\pm$.}
We have to compute
\begin{align*}
I_\mp(k_1,k_2) := \frac{i}{(2\pi)^4} 
& \int d^4k\, \th(\mp(k_1^0 - k^0)) \,\dl((k_1 - k)^2 - M^2)
\\
&\quad \x \frac{1}{k^2 - M^2 + i0}\, \th(\mp(k^0 + k_2^0))
\,\dl((k + k_2)^2 - M^2).
\end{align*}
Let us make a change of variable $q := k + k_2$, and introduce
$P := k_1 + k_2$, noting for later purposes that $P^2 = 2(k_1k_2)$.
One obtains the integral:
\begin{align}
\int d^4q\, \th(\mp(P^0 - q^0))\,\dl((P - q)^2 - M^2)\,
\frac{1}{(q - k_2)^2 - M^2 + i0}\, \th(\mp q^0) \,\dl(q^2 - M^2).
\label{eq:I-supp} 
\end{align}
It follows that $I_\mp(k_1,k_2) \propto\th(\mp P^0) \,\th(P^2 - 4M^2)$,
and that $\sgn k_1^0 = \sgn k_2^0$ for $P^2 \geq 4M^2$.

Performing the $q^0$-integration and using the notation
$E_q := \sqrt{|\qq|^2 + M^2}$, we extract
\begin{align*}
I_\mp(k_1,k_2) &= \frac{i}{(2\pi)^4} \,\th(\mp P^0) \,\th(P^2 - 4M^2)
\\
&\quad \x 
\int \frac{d^3q}{2E_q}\, \th(\mp(P^0 - q^0)) \,\dl((P - q)^2 - M^2)\, 
\frac{1}{(q - k_2)^2 - M^2 + i0}\biggr|_{q^0=\mp E_q}.
\end{align*}

Since $P^2 > 0$, one may choose a particular Lorentz frame such that
\begin{equation}
P = (P^0,\zero); \word{hence} \kk_1 = - \kk_2, \quad
k_1^0 = \mp|\kk_1| = \mp|\kk_2| = k_2^0 = \half P^0.
\label{eq:P=(P0,0)} 
\end{equation}
Taking into account $q^2 = M^2$, we observe that
$(P - q)^2 - M^2 = 2P^0(\half P^0 - q^0)$, which yields
$$
\dl((P - q)^2 - M^2) = \frac{\dl(q^0 - \half P^0)}{2|P^0|}
= \frac{\dl(E_q - \half|P^0|)}{2|P^0|}\,,
$$
by using $q^0 = \mp E_q$. For later aims, we point out that in the
chosen frame this distribution implies $q^0 = k_2^0$; hence
\begin{equation}
kP = (q - k_2)P = (q^0 - k_2^0)P^0= 0.
\label{eq:q0=k20} 
\end{equation}
From $\mp q^0 = \mp\half P^0$ comes $\mp(P^0 - q^0) = \mp\half P^0 >
0$. Therefore the factor $\th(\mp(P^0 - q^0))$ is redundant. Changing
the integration variables,
$$
\int d^3q \cdots = \int_M^\infty dE_q\, E_q\sqrt{E_q^2 - M^2} \int
d\Om_q \cdots\,,
$$
the $E_q$-integration can trivially be done, and we are left with:
\begin{align}
I_\mp(k_1,k_2) = i\,\th(\mp P^0)\, \th(P^2 - 4M^2)\,
\frac{\sqrt{(P^0)^2 - 4M^2}}{(2\pi)^4\,8|P^0|} \int\!\frac{d\Om_q}{(q
- k_2)^2 - M^2 + i0}\biggr|_{q^0=P^0/2}.
\label{eq:I1} 
\end{align}

Let $\al$ be the angle between $\kk_2$ and $\qq$, and let $z :=
\cos\al$. Due to $q^2 = M^2$, $k_2^2 = 0$, $|\qq| = \sqrt{E_q^2 - M^2}
= \half\,\sqrt{P_0^2 - 4M^2}$ and relations \eqref{eq:P=(P0,0)}
and~\eqref{eq:q0=k20}, we obtain
\begin{equation}
(q - k_2)^2 - M^2 = -2(k_2 q) = -2(k_2^0 q^0 - |\qq|\.|\kk_2|\,z)
= \tfrac{a}{2}(-a + bz),
\label{eq:denominator} 
\end{equation}
where
$$
a := |P^0| > 0,  \quad  0 \leq b := \sqrt{(P^0)^2 - 4M^2} < a.
$$
We point out that $(-a + bz) < 0$ for all $z \in [-1,1]$: there is no
infrared problem in our triangle graph. The remaining $\Om_q$-integral
can be easily computed:
\begin{equation}
\frac{4\pi}{a} \int_{-1}^1 \frac{dz}{-a + bz} =
\frac{4\pi}{|P^0|\,\sqrt{(P^0)^2 - 4M^2}}\, \log\frac{(P^0)^2 -
|P^0|\,\sqrt{(P^0)^2 - 4M^2} - 2M^2}{2M^2}\,.
\label{eq:in-aeternum} 
\end{equation}
To obtain the result in a generic Lorentz frame, replace 
$(P^0)^2$ by $s := P^2 = 2(k_1k_2)$, so
\begin{align}
I_\mp(k_1,k_2) &= \frac{i\th(\mp P^0)\,\th(s - 4M^2)}{4(2\pi)^3\,s}
\log\biggl[ \frac{s - \sqrt{s(s - 4M^2)}}{2M^2} - 1 \biggr]
\notag
\\
&=: \th(\mp P^0)\,\th(s - 4M^2)F(s).
\label{eq:scalar-I} 
\end{align}

The result for $J_\pm(k_1,k_2)$ can be read off from \eqref{eq:I1} by
omitting the Feynman propagator $i(2\pi)^{-2}\,((q - k_2)^2 - M^2 +
i0)^{-1}$. One obtains for the contribution of the $J$-integrals:
\begin{align*}
J_\pm(k_1,k_2) 
= \frac{1}{8\pi}\,\th(\pm P^0)\,\th(s - 4M^2)\, \sqrt{1 - 4M^2/s}\,.
\end{align*}

\paragraph{Vector integrals $I^\mu_\mp\,$.}
For the same reasons as for the scalar integral, it must hold that
$I^\mu_\mp(k_1,k_2) \propto \th(\mp P^0)\,\th(s - 4M^2)$. From Lorentz
covariance and $I^\mu_\pm(k_1,k_2) = - I^\mu_\pm(k_2,k_1)$ it follows
$$
I^\mu_\mp(k_1,k_2) 
= \th(\mp P^0)\, \th(s - 4M^2)\, (k_1^\mu - k_2^\mu)\,G(s)
$$
for appropriate $G(s)$. An immediate consequence is $I^\mu P_\mu = 0$.
To procure $G(s)$, compute
$$
k_{2,\mu}I^\mu_\mp(k_1,k_2) = \frac{1}{2} \th(\mp P^0)\, \th(s -
4M^2)\, s\,G(s) =\bigl(-i/8\,(2\pi)^3 \bigr)\,\th(\mp P^0)\, \th(s -
4M^2)\, \sqrt{1 - 4M^2/s}
$$
The second equality is obtained by comparing with the scalar integral:
there is an extra factor $(k_2 k) = (k_2 q) = -a(-a + bz)/4$, where
\eqref{eq:denominator} is used. Then the $\Om_q$-integral becomes
trivial. Thus we glean
\begin{align}
G(s) = \frac{-i}{32\,\pi^3\,s}\sqrt{1 - 4M^2/s}\,.
\label{eq:G} 
\end{align}

\paragraph{Tensor integrals $I^{\mu\nu}_\mp$.}
Proceeding analogously to the vector integrals, one argues that
$$
I^{\mu\nu}_\mp(k_1,k_2) = \th(\mp P^0)\,\th(s - 4M^2)\,
\bigl[ (k_1^\mu k_1^\nu + k_2^\mu k_2^\nu)\,A(s) 
+ (k_1^\mu k_2^\nu + k_2^\mu k_1^\nu)\,B(s) + g^{\mu\nu}\,C(s) \bigr].
$$
We need three independent identities to compute $A(s)$, $B(s)$
and~$C(s)$. A first one is:
\begin{align}
I^{\mu\nu}_\mp k_{2\mu} k_{2\nu} 
&= \th(\mp P^0)\, \th(s - 4M^2)\, A(s)\, s^2/4
\notag \\
&= \th(\mp P^0)\, \th(s - 4M^2)\, \frac{-i}{2^5\,(2\pi)^3}\,
s\,\sqrt{1 - 4M^2/s}\,.
\label{eq:Ikk} 
\end{align}
The second equality is obtained by a modification of the computation
of the scalar integral: there is the extra factor $(k_2 k)^2 = a^2(-a
+ bz)^2/16$. This yields $A(s)=G(s)/2$. A second identity is
given by the \textit{trace}. The result is again obtained by comparing
with the computation of the scalar integral: there is an additional
factor $k^2 = (q - k_2)^2 = M^2 - 2(k_2q) = M^2 - 2(k k_2)$, hence
$$
I^\mu_{\mp,\mu} = \th(\mp P^0)\,\th(s - 4M^2)\,(sB + 4C)
= M^2 I_\mp - 2k_{2,\mu} I^\mu_\mp.
$$
A third identity following from \eqref{eq:q0=k20} reads:
$$
I^{\mu\nu}_\mp P_\nu = \th(\mp P^0)\, \th(s - 4M^2)\, P^\mu\bigl((A +
B)s/2 + C\bigr) = 0.
$$
Pulling together these results, one arrives~at
$$
B(s) = - M^2 F(s)/s  \word{and}  C(s) = M^2 F(s)/2 - s\,G(s)/4.
$$

\marker
At this point we are able to show that the triangle plus fish-like
parts constitute a gauge-invariant quantity. For that, insert the
results already known for the integrals
into~\eqref{eq:suerte-o-verdad}, obtaining:
\begin{align}
& d^{\mu\nu}(k_1,k_2)\biggr|_{k_1^2=0=k_2^2} 
= \frac{\sgn(P^0)\,\th(s - 4M^2)}{(2\pi)^2}\,\bigl[
k_1^\mu k_2^\nu [4G(s) - (1 + 4M^2/s)F(s)]
\notag 
\\
&\quad + 2M^2 g^{\mu\nu} F(s) 
- k_1^\nu k_2^\mu \frac{4M^2}{s}\, F(s) \bigr]
= \sgn(P^0)\,\th(s - 4M^2) \frac{4M^2}{(2\pi)^2}\,
P^{\mu\nu}\,\frac{F(s)}{s}.
\label{eq:d-onshell-0} 
\end{align}
The $k_1^\mu k_2^\nu$-terms 
have been dropped in the last identity, due to $k_\mu A^\mu(-k)=0$. The
remainder is electromagnetically gauge-invariant. Introducing the
dimensionless variable
$$
\tilde\rho := \frac{s}{4M^2} = \frac{P^2}{4M^2}\,,
$$
keeping in mind formula \eqref{eq:scalar-I}, and on use of
\eqref{eq:consilia-non-sentis}, equation \eqref{eq:d-onshell-0} can be
rewritten as
\begin{gather}
d^{\mu\nu}_\gi(k_1,k_2)\biggr|_{k_1^2=0=k_2^2}
:= \frac{i\,\sgn(P^0)\,\th(\tilde\rho - 1)}{(2\pi)^5}\,
P^{\mu\nu}\, b(\tilde\rho)
\label{eq:d0-h-h1} 
\\
\shortintertext{with} 
b(\tilde\rho)
:= \frac{1}{16\,M^2\,\tilde\rho^2}\,
\log\bigl( 2\tilde\rho - 2\sqrt{\tilde\rho(\tilde\rho - 1)} - 1 \bigr)
= - \frac{1}{16\,M^2\,\tilde\rho^2}\, 
\log\frac{1 + \sqrt{1 - \tilde\rho^{-1}}}
{1 - \sqrt{1 - \tilde\rho^{-1}}} \,,
\notag
\end{gather}
where `gi' stands for the gauge invariant part. The singular order of
$d^{\mu\nu}_{\gi}\bigr|_{k_1^2=0=k_2^2}$ is $\om = -2$ by power
counting; whereas for the off-shell $d^{\mu\nu}(k_1,k_2)$ the value is
$\om = 0$.

\subsection{Regularity of absorptive parts in momentum space}
\label{ssc:do-ut-des}

This subsection is devoted to prove essential regularity properties of
the \textit{off-shell} $d$-distribution, more precisely of
$d^{\mu\nu}(k_1,k_2)$, for $(k_1,k_2) \in \sV := \bigl( \ovl V_+ \less
\{0\} \bigr)^{\x 2} \cup \bigl( \ovl V_- \less \{0\} \bigr)^{\x 2}$.
We look at the terms coming from \eqref{eq:suerte-o-verdad} by means
of \eqref{eq:def-J}. Introducing the new integration variable $q := -
k + \half(k_1 - k_2)$, the internal lines' momenta are
\begin{equation}
q_1 = q + \half P, \quad 
q_2 = q - \half P, \quad
q_3 = q - \half(k_1 - k_2),
\label{eq:q123} 
\end{equation}
and one sees that the considered terms are all of the type
\begin{align}
& H^{\mu\nu}(k_1,k_2)
\notag \\
&:= \int d^4 q\, \bigl(
\th(q_1^0)\, \th(-q_2^0) - \th(-q_1^0)\,\th(q_2^0) \bigr)
\,\dl(q_1^2 - M^2) \,\dl(q_2^2 - M^2)\,
\frac{h^{\mu\nu}(k_1,k_2,q)}{M^2 - q_3^2}
\label{eq:F_absorptive_gen} 
\end{align}
for $(k_1,k_2) \in \sV_1 := \bigl(\ovl V \less \{0\}\bigr)^{\x 2}$
with $\ovl V := \ovl V_+ \cup \ovl V_-$, and where $h^{\mu\nu}\colon
\bR^{4\x 3} \to \bC\,$ is a polynomial of degree~$2$. We have used that
for $(k_1,k_2) \in \sV_1$ it holds true that
\begin{equation}
\int d^4 q\, \bigl( \th(q_1^0)\,\th(-q_2^0) 
- \th(-q_1^0)\,\th(q_2^0) \bigr) 
\,\dl(q_1^2 - M^2) \,\dl(q_2^2 - M^2) \,\dl(q_3^2 - M^2) = 0.
\label{eq:dldldl=0} 
\end{equation}
This last relation can be argued as follows:%
\footnote{We borrow the standard notation for the mass shell: 
$H_M^\pm := \set{p \in \bR^4 : p^2 = M^2,\ \pm p_0 > 0}$.}
$\!\!$the various $\th$- and $\dl$-distributions yield the
restrictions $(q_1,q_2) \in (H_M^+ \x H_M^-) \cup (H_M^- \x H_M^+)$
and $q_3 \in H_M^+ \cup H_M^-$; taking moreover into account that $q_3
= q_2 + k_2$ and $q_3 = q_1 - k_1$, it ensues that the various
restrictions on $q_3$ are not compatible.

The same identity implies that terms of the kind 
$T_1(x_{\pi 1})\,T_1(x_{\pi 2})\,T_1(x_{\pi 3})$ do not contribute to
the third commutator in formula \eqref{eq:para-llorar} for~$D_3$ when
$(k_1,k_2) \in \sV_1$, for all permutations~$\pi$: the \textit{whole}
contribution to $d^{\mu\nu}(k_1,k_2)|_{(k_1,k_2)\in\sV_1}$ coming from
this commutator is of the kind~\eqref{eq:F_absorptive_gen}.

The contributions to $d^{\mu\nu}(k_1,k_2)|_{(k_1,k_2)\in \sV}$ coming
from the other two commutators in \eqref{eq:para-llorar} are of the
same form up to cyclic permutations 
$k_1 \mapsto k_2 \mapsto -(k_1 + k_2) \mapsto k_1$ of the external
momenta. Here we use that $(k_1, k_2)\in \sV$ implies 
$(k_2, -k_1 - k_2) \in \sV_1$ and $(-k_1 - k_2, k_1) \in \sV_1$, hence
we may apply the identity \eqref{eq:dldldl=0} also for the permuted
momenta. However, note that the polynomials $h_j^{\mu\nu}$, $j = 2,3$,
belonging to these other two cuts are not obtained by cyclic
permutations of the external momenta in the original polynomial
$h_1^{\mu\nu}$, meant in \eqref{eq:F_absorptive_gen}. This is due to
the difference between the higgs vertex and the photon vertices; in
particular, these other two cuts contain no term giving rise to a
fish-like diagram. Summing up, it holds that
\begin{equation}
d^{\mu\nu}(k_1, k_2)\bigr|_{(k_1,k_2)\in\sV}
= H^{\mu\nu}_1(k_1, k_2) + H^{\mu\nu}_2(k_2, -k_1 - k_2)
+ H^{\mu\nu}_3(-k_1 - k_2, k_1),
\label{eq:dV1} 
\end{equation}
for some $H^{\mu\nu}_j$ $(j = 1,2,3)$ of the
form~\eqref{eq:F_absorptive_gen}, the pertinent polynomials
$h_j^{\mu\nu}$ being of degree~$2$.

\begin{lema} 
\label{lem:continuity_absorptive_part}
Let $q_1,q_2,q_3$ and $\sV_1$ be defined as above in~\eqref{eq:q123}
and after~\eqref{eq:F_absorptive_gen}, and let $H^{\mu\nu} \colon
\bR^{4\x2} \to \bC$ be given in terms of a generic polynomial
$h^{\mu\nu} \colon \bR^{4\x3} \to \bC$ of degree $\zeta \in \bN_0$, as
in~\eqref{eq:F_absorptive_gen}. Then for all $(k_1,k_2) \in \sV_1$ and
for some $C > 0$ the function $H^{\mu\nu}$ is \textbf{continuous} in
the region~$\sV_1$, and can be bounded as follows:
\begin{equation}
|H^{\mu\nu}(k_1,k_2)| 
\leq C\,\frac{(1 + |(k_1,k_2)|)^\zeta}{|(k_1k_2)|}\,
\th((k_1 + k_2)^2 - 4M^2)\, \log((k_1 + k_2)^2/M^2).
\label{eq:function_F_lemma_bound} 
\end{equation}
Note that $|(k_1 k_2)| > 0$ if $(k_1,k_2) \in \sV_1$ and 
$(k_1 + k_2)^2 \geq 4M^2$.
\end{lema}

\begin{proof}
Let $P := k_1 + k_2$ and $k := k_1 - k_2$. We first observe, on the
strength of
\begin{gather*}
q_1^2 - q_2^2 = 2(Pq), \quad 
q_1^2 + q_2^2 - 2M^2 = 2(q^2 + \quarter P^2 - M^2)
\\
\word{and of\!\!} M^2 - q^3 = (M^2 - q^2 - \quarter P^2) + \quarter
P^2 - \quarter k^2 + (k q) \word{\!\!that}
\\
H^{\mu\nu}(k_1,k_2) \sim \sgn(P^0) \int d^4q\, \dl(q^2 + \quarter P^2
- M^2) \,\dl((Pq))\, \frac{h^{\mu\nu}(k_1,k_2,q)}{P^2/4 - k^2/4 +
(kq)}\,,
\end{gather*}
omitting irrelevant prefactors. Since $q_1 - q_2 = P$ and 
$(q_1,q_2) \in (H_M^+ \x H_M^-) \cup (H_M^- \x H_M^+)$, we know that
$H^{\mu\nu}(k_1,k_2)$ vanishes for $P^2 < 4M^2$. Hence, to perform the
integrals in $q^0$ and~$|\qq|$ using the Dirac deltas, we may work in
the frame in which $\PP = 0$. There the $\dl$-distributions yield
$q^0 = 0$ and $|\qq| = \sqrt{P_0^2/4 - M^2}$.

With the notation $\hat p := \pp/|\pp|$ for $p \in \set{\!q,k\!}$, it
follows that
$$
\quarter P^2 - \quarter k^2 + (kq) 
= (k_1k_2) \biggl( 1 - (\hat q\,\hat k) \sqrt{1 - 4M^2/P^2}\,
\frac{|P^0|\,|\kk|}{2(k_1 k_2)} \biggr),
$$
and one verifies that
\begin{equation}
0 \leq P_0^2\,|\kk/2|^2 = (k_1 k_2)^2 - k_1^2 k_2^2,
\label{eq:N} 
\end{equation}
with $k_1 = (k_1^0,\kk_1)$ and $k_2 = (P^0 - k_1^0, -\kk_1)$. With the
help of these results we obtain
\begin{align}
(k_1 k_2)\, H^{\mu\nu}(k_1,k_2) 
&\sim \sgn(P^0)\, \th(P^2 - 4M^2)\, \sqrt{1 - 4M^2/P^2}
\notag \\
&\enspace \x \int_{\bS^2} d\Om(\hat q)\,
\frac{h^{\mu\nu}\bigl(k_1, k_2, (0,\sqrt{P^2/4 - M^2}\,\hat q)\bigr)}
{1 - (\hat q\,\hat k) 
\sqrt{(1 - 4M^2/P^2)\,(1 - k_1^2 k_2^2/(k_1 k_2)^2)}} \,,
\label{eq:H} 
\end{align}
valid in the frame in which $\PP = \zero$. Let moreover $\sV_1^M :=
\set{(k_1,k_2) \in \sV_1 : (k_1 + k_2)^2 \geq 4M^2}$. We know that
\begin{gather*}
4M^2/P^2 \in (0,1] \word{and} k_1^2 k_2^2/(k_1 k_2)^2 \in [0,1]
\word{for} (k_1,k_2) \in \sV_1^M;
\\
\word{hence} 
a := \sqrt{(1 - 4M^2/P^2)\,(1 - k_1^2 k_2^2/(k_1 k_2)^2)} \in [0,1).
\end{gather*}
In particular, the denominator in the integrand of~\eqref{eq:H} does
not vanish for $(k_1,k_2) \in \sV_1^M$. Since $\th(P^2 - 4M^2)\,
\sqrt{1 - 4M^2/P^2}$ is continuous, $H^{\mu\nu}$ \textit{is
continuous} on~$\sV_1$.

Observe now that for all $\hat q \in \bS^2$ the inequality
$$
\bigl| h^{\mu\nu}\bigl( 
k_1, k_2, (0,\sqrt{P^2/4 - M^2} \hat q) \bigr) \bigr|
\leq \const (1 + |(k_1, k_2)|)^\zeta
$$
holds, with $|(k_1,k_2)|^2 := \sum_{j=0}^3 (k_{1j}^2 + k_{2j}^2)$.

Setting $z := \bigl(\hat q\hat k\bigr)$, the remaining integral is of
the type
$$
\int_{-1}^1 \frac{dz}{1 - az} 
= \frac{1}{a} \log\biggl( \frac{1 + a}{1 - a} \biggr) 
\leq 2(1 - \log(1 - a)),
$$
valid for $a \in [0,1)$. Using that $a \leq \sqrt{1 - 4M^2/P^2} \leq
(1 - 2M^2/P^2)$ and monotonicity of the logarithm, we see that
$$
-\log(1 - a) = \log \frac{1}{1 - a} \leq \log \frac{P^2}{2M^2} \,.
$$
Putting together the estimates, we end up with
\begin{equation}
|(k_1 k_2)\, H^{\mu\nu}(k_1,k_2)| 
\leq \const \. \th(P^2 - 4M^2)\, (1 + |(k_1,k_2)|)^\zeta
\bigl( 1 + \log(P^2/2M^2) \bigr),
\label{eq:prove_continuity_absorptive} 
\end{equation}
impliying \eqref{eq:function_F_lemma_bound}, since $1 + \log(P^2/2M^2)
< 2\,\log(P^2/M^2)$ for $P^2 \geq 4M^2$.
\end{proof}

The reader should keep in mind that $d^{\mu\nu}(k_1,k_2)$ is supported
outside a certain neighbourhood of the origin on momentum space --
have a look back at Eq.~\eqref{eq:H}.

\begin{corl} 
\label{cor:d-regular}
The off-shell $d$-distribution $d^{\mu\nu}(k_1,k_2)$ given in
\eqref{eq:dV1} is continuous on~$\sV$ and fulfills the bound:
\begin{equation}
|d^{\mu\nu}(k_1,k_2)| 
\leq \const\, \frac{(1 + |(k_1,k_2)|)^{\om+2}}{|(k_1k_2)|}\,
\log(2 + |(k_1,k_2)|/M) \text{ for all } (k_1,k_2) \in \sV.
\label{eq:bound_d} 
\end{equation}
\end{corl}

\begin{proof}
Continuity follows immediately from Lemma
\ref{lem:continuity_absorptive_part}. For the bound \eqref{eq:bound_d}
we have substituted $\om + 2 \equiv \om(d) + 2$ for $\zeta$ of the
Lemma, since the singular order of $H_j^{\mu\nu}$ ($j = 1,2,3$) is
$\zeta - 2$ by power counting in \eqref{eq:F_absorptive_gen}. In
addition, for $H^{\mu\nu}_1(k_1,k_2)$ we have used that $(k_1 + k_2)^2
\leq 4\,|(k_1,k_2)|^2$, and in order to omit the $\th$-distribution we
have replaced $\log(2\,|(k_1,k_2)|/M)$ by $2\log(2 + |(k_1,k_2)|/M)$.
One deals analogously with $H^{\mu\nu}_2(k_2, -k_1 - k_2)$ and
$H^{\mu\nu}_3(-k_1 - k_2, k_1)$.
\end{proof}

\subsection{Distribution splitting by the dispersion integral for null
momenta}
\label{ssc:hasta-ahi-podiamos-llegar}

Recall that for $(k_1,k_2)\in V_\eta \x V_\eta$ the advanced part
$a^{\mu\nu}$ of $d^{\mu\nu}$ can be computed by the dispersion
integral \eqref{eq:hoist-with-retard}. Using the regularity properties
of $d^{\mu\nu}$ given in Corollary \ref{cor:d-regular}, we finally aim
to show that the limit $k_1^2 \downto 0$, $k_2^2 \downto 0$ in
\eqref{eq:hoist-with-retard} commutes with integration; that is, the
dispersion integral is also valid for $k_1^2 = 0 = k_2^2$. To
formulate the assertion, let
\begin{equation}
\sK := \set{(k_1,k_2) \in (\bR^4)^{\x 2} 
: k_1^2, k_2^2 < 4M^2, \ (k_1 + k_2)^2 < 4M^2, \ (k_1 k_2)\neq 0}.
\label{eq:valet} 
\end{equation}

Bearing in mind the factors $\th(q^2 - 4M^2)$ for $q \in \{k_1, k_2,
k_1 + k_2\}$ appearing in each term of $d^{\mu\nu}(k_1,k_2)$, we see
that for $(k_1,k_2) \in (V_\eta \x V_\eta) \cap \sK$, formula
\eqref{eq:hoist-with-retard} can be rewritten as:
\begin{equation}
a^{\mu\nu}(k_1,k_2) = \frac{i\eta}{2\pi} \int_{|t|\geq t_{\min}} dt\,
\frac{d^{\mu\nu}(tk_1, tk_2)}{t^{\om+1}(1 - t)},
\label{eq:dispersion_simplified} 
\end{equation}
for some $t_{\min} > 1$ depending on $k_1,k_2$. Now, as discussed in
subsection \ref{sec:ipso-facto}, one knows $a^{\mu\nu}(k_1,k_2)$ to be
analytic on the region $\sK$. The Lebesgue dominated convergence
theorem~\cite[Th.~4.6.3]{RealKippa} with the bound~\eqref{eq:bound_d}
allows us conclude that \eqref{eq:dispersion_simplified} is a valid
identity for $(k_1,k_2) \in \sV \cap \sK$. Indeed, introducing the set
of limit points
$$
\sM := \sV \cap \sK \cap \{(k_1,k_2) \in \bR^8 : k_i^2 = 0\}
= \{(k_1,k_2) \in \bR^8 : k_i^2 = 0, \ 0 < (k_1 + k_2)^2 < 4M^2\},
$$
it is enough to observe that for any $(\tilde k_1,\tilde k_2) \in \sM$
-- implying $(\tilde k_1\,\tilde k_2) > 0$ and
$\tilde k_1^0 \tilde k_2^0 > 0$ -- there is a neighbourhood
$\sU_{(\tilde k_1,\tilde k_2)}$ such that
\begin{align*}
\biggl| \th(|t|-t_{\min}) \frac{d(tk_1,tk_2)}{t^{\om+1}(1-t)} \biggr|
&\leq \const \. \frac{\th(|t| - t_{\min})}{|t(1 - t)|}\,
\frac{\bigl( 1 + |(k_1,k_2)| \bigr)^{\om+2}}{|(k_1\,k_2)|}\,
\log(2 + |t|\,|(k_1,k_2)|/M)
\\
&\leq C\,\frac{\th(|t| - t_1)}{|t(1 - t)|}\, \log(2 + C_1\,|t|),
\end{align*}
for all $(k_1,k_2) \in (V_\eta \x V_\eta) \cap \sK 
\cap \sU_{(\tilde k_1,\tilde k_2)}\,$, for some $C,C_1 > 0$ and some
$t_1 > 1$ independent of $(k_1,k_2)$. The function on the right hand
side is absolutely integrable in~$t$ -- here we see the reason for the
condition $(k_1k_2) \neq 0$ in~\eqref{eq:valet}.

\subsection{Normalization of the scalar model by distribution
splitting}
\label{ssc:splitting}

We must finally compute the gauge invariant part
$t^{\mu\nu}_\gi(k_1,k_2)$ for momenta lying on the set~$\sM$.
Considering the formula $T_{3} = A_{3} - A'_{3}$ and reckoning that
$a'^{\mu\nu}(k_1,k_2)|_{k_1^2=0=k_2^2}$ contains the factor $\th(P^2 -
4M^2)$ where $P := k_1 + k_2$, we see that on~$\sM$ its contribution
vanishes, that is $t^{\mu\nu} = a^{\mu\nu}$ there.

The upshot of the preceding two subsections is that we may compute
valid terms of the central solution $a^{\mu\nu}|_\sM \equiv
a^{c\,\mu\nu}|_\sM$ by inserting the on-shell amplitude
\eqref{eq:d-onshell-0} into the dispersion integral, with $\om$ the
singular order \textit{of the off-shell} $d^{\mu\nu}$, equal to~$0$ in
the present case.

Looking at \eqref{eq:d-onshell-0}, observe that a $k_r^\mu k_s^\nu$-
or $g^{\mu\nu}$-term of~$d^{\mu\nu}$ goes over to a $k_r^\mu k_s^\nu$-
or $g^{\mu\nu}$-term of~$a^{\mu\nu}$, respectively. Therefore, such
factors may be taken out of the dispersion integral. Since moreover
$P^{\mu\nu}(tk) = t^2\,P^{\mu\nu}(k)$, we see that the gauge invariant
part $a^{\mu\nu}_\gi$ can be obtained by inserting just the gauge
invariant part $d^{\mu\nu}_\gi$ in~\eqref{eq:d0-h-h1} into the
dispersion integral. The latter is of the form \eqref{eq:d-form}. So
we may use the version \eqref{eq:pro-reo} of the dispersion integral.

Lastly, $t^{\mu\nu}_\gi|_\sM = a^{\mu\nu}_\gi|_\sM$ is obtained from
\eqref{eq:pro-reo} by setting $\om = 0$ and substituting there
$b(u\tilde\rho)$ as given in~\eqref{eq:d0-h-h1} for
$f\bigl(u(k_1^2, k_2^2, (k_1 + k_2)^2)\bigr)$ -- in our case only
$(k_1 + k_2)^2 = s$ is present. Allowing for the dilation factor in
$P^{\mu\nu}$ this leads, for $(k_1,k_2) \in \sM$, to
\begin{align}
t^{\mu\nu}_\gi(k_1,k_2) 
&= - \frac{P^{\mu\nu}}{(2\pi)^6} \int_{\tilde\rho^{-1}}^\infty du\,
\frac{u\,b(u\tilde\rho)}{u(1 - u)}
\label{eq:socratic} 
\\
&= \frac{P^{\mu\nu}}{16 M^2\,(2\pi)^6} \int_{\tilde\rho^{-1}}^\infty
\frac{du}{\tilde\rho^2 u^2(1 - u)}\,
\log \frac{1 + \sqrt{1 - u^{-1}\tilde\rho^{-1}}}
{1 - \sqrt{1 - u^{-1}\tilde\rho^{-1}}}
= - \frac{P^{\mu\nu}}{8\,(2\pi)^6}\, \frac{J_2(\tilde\rho)}{M^2}\,,
\notag \\
\shortintertext{where}
2\,J_2(\tilde\rho) 
&:= \int_1^\infty dv\, \frac{1}{(v - \tilde\rho)v^2}\,
\log \frac{1 + \sqrt{1 - v^{-1}}}{1 - \sqrt{1 - v^{-1}}},
\notag
\end{align}
after the change of integration variable $v := u \tilde\rho$.
Integrals like $J_2$ have been computed in \cite{BoraChristovaEberl}.
From Appendix~C of that reference:
\begin{equation}
J_1(\tilde\rho,a) 
:= \frac{1}{2} \int_1^\infty dv\, \frac{1}{(v - \tilde\rho)(v - a)}\,
\log \frac{1 + \sqrt{1 - v^{-1}}}{1 - \sqrt{1 - v^{-1}}}
= \frac{f(\tilde\rho) - f(a)}{\tilde\rho - a}
\label{eq:J1} 
\end{equation}
for $0 \leq \tilde\rho \leq 1$, $0 \leq a \leq 1$, where $f$ is the
distribution \eqref{eq:f1}. We infer that
\begin{align}
J_2(\tilde\rho) = \frac{\del}{\del a}\biggr|_{a=0}\, J_1(\tilde\rho,a)
= \frac{f(\tilde\rho)}{\tilde\rho^2} - \frac{1}{\tilde\rho}
\word{for} 0 \leq \tilde\rho \leq 1,
\label{eq:J2} 
\end{align}
by bringing in the values $f(0) = 0$ and $f'(0) = 1$, which can be
read off from~\eqref{eq:arcsin-2}.

Summing up, the final result reads, as expected:
\begin{align}
t^{\mu\nu}_\gi(k_1,k_2) 
= \frac{P^{\mu\nu}}{8\,(2\pi)^6}\, \frac{1}{M^2} \biggl(
\frac{1}{\tilde\rho} - \frac{f(\tilde\rho)}{\tilde\rho^2} \biggr)
= \frac{P^{\mu\nu}}{8\,(2\pi)^6}\, \frac{F_0(\tilde\rho)}{M^2}
\quad\text{for } (k_1,k_2) \in \sM,
\label{eq:F0} 
\end{align} 
where $F_0$ was given in~\eqref{eq:magister-dixit}. We conjecture that
this formula holds true for all $(k_1,k_2)$ satisfying 
$k_1^2 = 0 = k_2^2$ and $(k_1 + k_2)^2 > 0$.

The reader should remember that \eqref{eq:F0} stands in principle for
just a member of a solution set. Since $\om = 0$, the general
Lorentz-invariant Epstein--Glaser solution is obtained by adding to
expression~\eqref{eq:F0} a term of the type $Cg^{\mu\nu}$ with 
$C \in \bC$ arbitrary. But such a term with $C \neq 0$ would violate
EGI. Therefore we regard the above result as unique.

Recovering formula \eqref{eq:dAAh} and the factor $4ge^2M$, one ends
up with
\begin{align*}
\int dx_1\,dx_2\,dx_3\, T_3(x_1,x_2,x_3) 
= \frac{g\al}{(2\pi)^3 M} \int dk_1\,dk_2\, 
h(k_1 + k_2) A^\mu(-k_1) A^\rho(-k_2)\, P_{\mu\nu}\, F_0(\tilde\rho),
\end{align*}
which, on substituting $\rho$ for~$\tilde\rho$, that is, $m_h^2$ for
$s \equiv (k_1 + k_2)^2$, agrees with the
literature~\cite{HungryHunters}.

\begin{remk} 
In the occasion an (unsubtracted) dispersion integral applied to
$b(u)$, performed in \cite[Eq.~(3.2)]{ChristovaI}, leads to the
\textit{same} integral \eqref{eq:socratic} and so the same correct
result. As the next section shows, this does not hold for the
higgs to diphoton decay via EW vector bosons.
\end{remk}

\section{Higgs to diphoton decay via EW vector bosons}
\label{sec:soberbia-pagana}

\subsection{Derivation of the quartic $AAWW^\7$-vertex in the unitary
gauge}
\label{ssc:per-angusta}

The amplitude in question in this paper describes an EW decay process
at third order in the coupling constant. Its structure is given by the
\textit{cubic} vertices in the first TOP~$T_1$ -- that is the
\textit{sole} ``empirical'' input. Here in going from $T_1$ to~$T_2$
we \textit{derive} the $AAWW^\7$-vertex which contributes by a
``fish-like'' diagram to the amplitude to be computed, see
Figure~\ref{fg:diagrams}.

The general idea is to examine the propagator which is to become the
internal line linking the di-photon in the one-loop, three-vertex
graph, and to obtain the one-loop, two-vertex graph from a
modification of that propagator, demanded by EGI -- by which here we
precisely understand invariance of the $\bS$-matrix under the
variations $A^\mu(x) \to A^\mu(x) + \del^\mu\La(x)$: interaction
dictates symmetry. The method is similar to the derivation of the
$AA\vf\vf^\7$ ``seagull'' vertex from the cubic coupling in scalar
QED, first performed in this way in~\cite{DKS93}.

The concept works on configuration space, as follows. Recall the
pertinent Hermitian vertex -- see for instance
\cite[Sect.~7.2.2]{Langacker}, explicitly referring to the unitary
gauge. With $G_{\mu\nu} := \del_\mu W_\nu - \del_\nu W_\mu$, one has:
\begin{align}
T_1(x_1) = ie[(W^\mu G^\7_{\mu\nu} - W^{\7\mu}G_{\mu\nu}) A^\nu
- W^\mu W_\nu^\7 F_\mu^{\,\nu}](x_1).
\label{eq:hic-iacet} 
\end{align}
All indicated operator products are Wick products. We copy a second
vertex similar to~\eqref{eq:hic-iacet}:
\begin{align}
T_1(x_2) = ie[(W^\rho G^\7_{\rho\la} - W^{\7\rho} G_{\rho\la}) A^\la
- W^\rho W_\la^\7 F_\rho^{\,\la}](x_2),
\label{eq:lepus} 
\end{align}
and proceed to make the contractions, according to the BEG method to
construct (the relevant sector of) the second-order $T_2(x_1,x_2)$ out
of $T_1(x_1),T_1(x_2)$. Now $T_2(x_1,x_2) = T_1(x_1)T_1(x_2)$ for
$x_1$ not to the past of~$x_2$, and $T_2(x_1,x_2) = T_1(x_2)T_1(x_1)$
for $x_2$ not to the past of~$x_1$.

In view of the triangle diagram given in Figure~\ref{fg:diagrams}, we
are only interested in contractions yielding a connected tree diagram
with the photons uncontracted. Once they are made, we analyze the
terms for which the resulting distributions are not uniquely defined
on the diagonal $x_1 = x_2$. The ambiguities in this extension are
eliminated by EGI as sole selection criterion.

From the last pair of terms in \eqref{eq:hic-iacet} and
\eqref{eq:lepus}, and with $\Dl^{\mu\bt}$ standing for the
Feynman--Proca propagator~\eqref{eq:cum-grano-salis} for the
$W$-bosons, there comes:
\begin{align}
& -e^2\, \bigl( F_{\mu\nu}(x_1) W^{\7\nu}(x_1) 
\vev{\T W^\mu(x_1)\, W^{\7\bt}(x_2)} W^\al(x_2) F_{\al\bt}(x_2)
\notag \\
&\quad + F_{\mu\nu}(x_1) W^{\mu}(x_1)
\vev{\T W^{\7\nu}(x_1)\, W^\al(x_2)} W^{\7\bt}(x_2) F_{\al\bt}(x_2)
\bigr)
\notag \\
&= -e^2\, F_{\mu\nu}(x_1) F_{\al\bt}(x_2) W^{\7\nu}(x_1)
\,\Dl^{\mu\bt}(x_1 - x_2) W^\al(x_2) + (x_1 \otto x_2),
\notag \\
&\qquad \text{denoted $T_2^D(x_1,x_2)$ for later use;}
\label{eq:TD(1,2)} 
\end{align}
where explicitly $\Dl_{\bt}^\mu = -\bigl( g_\bt^\mu +
\del^\mu\del_\bt/M^2 \bigr)\Delta^F$ with $\Delta^F$ the scalar
Feynman propagator~\eqref{eq:accidit-in-puncto}. Electromagnetic gauge
variations obviously do not affect~\eqref{eq:TD(1,2)}.

By pairing the last term in \eqref{eq:lepus} with $A$-type
terms in \eqref{eq:hic-iacet} and the last term in
\eqref{eq:hic-iacet} with $A$-type terms in~\eqref{eq:lepus} we obtain 
eight terms, that we organize as follows:
\begin{align}
& e^2\, A^\mu(x_1) F_{\rho\bt}(x_2) \bigl[
G^\7_{\al\mu}(x_1) \,\Dl^{\al\bt}(x_1 - x_2)\, W^\rho(x_2)
- G_{\al\mu}(x_1) \,\Dl^{\al\rho}(x_1 - x_2)\, W^{\bt\7}(x_2) +{}
\notag \\
& W^\al(x_1) (-g^\rho_\mu \del_\al + g^\rho_\al \del_\mu)
\Delta^F(x_1{-}x_2) W^{\bt\7}(x_2) 
- W^{\al\7}(x_1) (-g^\bt_\mu \del_\al + g^\bt_\al \del_\mu) 
\Delta^F(x_1{-}x_2) W^{\rho}(x_2) \bigr]
\notag \\
& + (x_1 \otto x_2) =: T_2^B(x_1,x_2) + T_2^C(x_1,x_2),
\label{eq:TB(1,2)} 
\end{align}
where $T_2^C$ denotes the $(x_1 \otto x_2)$-term. As noted in
App.~\ref{app:spin-one-basics}, third-order derivatives of~$\Delta^F$
cancel here of their own accord. In order to verify EGI in the above
expression, note first that it can only be violated on the diagonal
$x_1 = x_2$, that is by contact terms -- see the discussion on this in
Sect.~\ref{ssc:lost-in-translation}. Therefore in computing the
divergence of expressions like the above one selects only such terms
that under $\del_1^\mu$ produce local expressions. For instance, the
third term in~\eqref{eq:TB(1,2)} brings forth:
$$
W^\al(x_1)\, g^\rho_\al\,\del_1^\mu\del_\mu \,\Delta^F(x_1 - x_2)\,
W^{\bt\,\7}(x_2) 
= -i W^\rho(x_1) W^{\7\bt}(x_1) \,\dl(x_1 - x_2) +\cdots
$$
where the dots stand for a term ${}\sim M^2 W \Delta^F W^\7$; but it
is not hard to see that it is cancelled by a similar one coming from
the next term. In conclusion: $T_2^B$ is individually
electromagnetically gauge-invariant, and $T_2^D$, $T_2^B$, $T_2^C$
calculated up to now have no bearing on the generation of the
$AAWW^\7$-vertex in constructing~$T_2$.

Still, we are left with the most interesting contractions to
calculate. By pairing the two first terms in~\eqref{eq:hic-iacet} with
the two first in~\eqref{eq:lepus} we get:
\begin{align}
& -e^2\,A^\mu(x_1) A^\rho(x_2) \bigl[
- W^\al(x_1) \wt D_{\al\mu\;\bt\rho}(x_1 - x_2) W^{\bt\7}(x_2)
- G^\7_{\al\mu}(x_1) \,\Dl^{\al\bt}(x_1 - x_2)\, G_{\bt\rho}(x_2) +{}
\notag \\
& G^\7_{\al\mu}(x_1) (g^\al_\rho \del_\bt - g^\al_\bt \del_\rho)
\Delta^F(x_1 - x_2) W^\bt(x_2) 
+ W^{\al\7}(x_1) (-g^\bt_\mu \del_\al + g^\bt_\al \del_\mu)
\Delta^F(x_1 - x_2) G_{\bt\rho}(x_2) \bigr]
\notag \\
& + (x_1 \otto x_2) =: T_2^A(x_1,x_2).
\label{eq:TA(1,2)} 
\end{align}
Outside the diagonal the distribution $\wt D_{\al\mu\;\bt\rho}$
necessarily coincides with $D_{\al\mu\;\bt\rho}$, defined in
Eq.~\eqref{eq:rarior-albo} as the propagator for the $G$-fields.
Following the Epstein--Glaser method, we now look for the most general
Lorentz-covariant extension of this distribution having the same
scaling degree. The solution reads:
\begin{equation}
\wt D_{\al\mu,\bt\rho}(y) = D_{\al\mu,\bt\rho}(y)
+ i(a\,g_{\al\bt} g_{\mu\rho} - b\,g_{\al\rho} g_{\mu\bt}) \,\dl(y)
\label{eq:ansatz-tildeD} 
\end{equation}
with as yet unknown numbers $a,b \in \bC$. Note that the second term
in the above propagator generates a contact term eventually yielding
the $AAWW^\7$-vertex
\begin{equation}
T_F(x_1,x_2) := e^2\,A^\mu(x_1) A^\rho(x_2) W^\al(x_1) W^{\7\bt}(x_2)
\, i(a\,g_{\al\bt} g_{\mu\rho} - b\,g_{\al\rho} g_{\mu\bt})
\,\dl(x_1-x_2).
\label{eq:sic-transit} 
\end{equation}
A third Lorentz tensor might appear in the Ansatz
\eqref{eq:ansatz-tildeD}, namely $g_{\al\mu} g_{\rho\bt} \,\dl(y)$.
However, on insertion into \eqref{eq:TA(1,2)}, one obtains the same
contribution as $g_{\al\rho}g_{\mu\bt}\,\dl(y)$, namely
$$
e^2\,A^\mu(x_1) A^\rho(x_1) W_\mu(x_1) W^{\7}_\rho(x_1)
\,\dl(x_1 - x_2).
$$

Since EGI of $T_2$ can be violated only by local terms, it will
suffice to select the terms which are $\sim (\del)\dl(x_1 - x_2)$
after taking the divergence $\del_1^\mu$. Those can be of two types:
\begin{enumerate}
\item 
Either the contact term already contains a $\dl(x_1 - x_2)$; or
\item 
Such terms are generated when the $\del^\mu_{x_1}$ acts on
$\del_\mu \Delta^F(x_1 - x_2)$ or on 
$\del\del_\mu \Delta^F(x_1 - x_2)$, due to
$(\square + m^2)\Delta^F = -i\dl$.
\end{enumerate}

From the $W(x_1)\,W^\7(x_2)$ part in \eqref{eq:TA(1,2)}, we find:
\begin{itemize}
\item 
The type~(a) contributions:
\begin{align}
& i \del^\mu W^\al(x_1)\,(-a\,g_{\al\bt} g_{\mu\rho} 
+ b\,g_{\al\rho} g_{\mu\bt}) \,\dl(x_1 - x_2)\, W^{\bt\,\7}(x_2)
\notag \\
&\quad
+ i W^\al(x_1)\,\bigl( (-a\,g_{\al\bt} \del_\rho 
+ b\,g_{\al\rho} \del_\bt)\dl \bigr)(x_1 - x_2)\, W^{\bt\,\7}(x_2)
\notag \\
&= i \bigl( -a\,\del_\rho W^\al(x_1) W^{\7}_\al(x_1)
+ b\,\del_\bt W_\rho(x_1) W^{\bt\7}(x_1) \bigr) \dl(x_1 - x_2)
\notag \\
&\quad
-ia\, W^\al(x_1) W^\7_\al(x_2) \,\del_\rho\dl(x_1 - x_2)
+ ib\, W_\rho(x_1 )W^{\bt\7}(x_2) \,\del_\bt\dl(x_1 - x_2),
\label{eq:WWa} 
\end{align}
and the type~(b) contribution:
\begin{align}
& i W^\al(x_1)\,\bigl(
(g_{\al\bt} \del_\rho - g_{\al\rho} \del_\bt)\dl \bigr)(x_1 - x_2)\,
W^{\bt\,\7}(x_2)
\notag \\
&= i W^\al(x_1) W^\7_\al(x_2) \,\del_\rho\dl(x_1 - x_2)
- i W_\rho(x_1) W^{\bt\7}(x_2) \,\del_\bt\dl(x_1 - x_2).
\label{eq:WWb} 
\end{align}

\item
From the $W^\7(x_1)\,W(x_2)$-part of the exchange term, we obtain
terms \eqref{eq:WWa} and~\eqref{eq:WWb} with $W \otto W^\7$
interchanged.

\item
From the $W^\7(x_1)\,G(x_2)$-part of the original term, we get the
type~(b) contribution:
\begin{align}
& i W^{\al\7}(x_1)\, g^\bt_\al \,\dl(x_1 - x_2)\, \bigl(
-\del_\bt W_\rho(x_2) + \del_\rho W_\bt(x_2) \bigr)
\notag \\
&= i \bigl( -W^{\al\7}(x_1)\,\del_\al W_\rho(x_1)
+ W^{\al\7}(x_1)\,\del_\rho W_\al(x_1) \bigr) \,\dl(x_1 - x_2).
\label{eq:WGb}
\end{align}

\item
From the $W(x_1)\,G^\7(x_2)$-part of the exchange term, we obtain the
term~\eqref{eq:WGb} with $W \otto W^\7$ interchanged.
\end{itemize}

Summing up, the requirement is:
\begin{align}
0 &\overset{!}{=}
W^\al(x_1) W^{\7}_\al(x_2) \,\del_\rho\dl(x_1 - x_2)\, (1 - a)
+ W_\rho(x_1) W^{\bt\7}(x_2) \,\del_\bt\dl(x_1 - x_2)\, (b - 1)
\notag \\
&\quad
+ \del_\bt W_\rho(x_1) W^{\bt\7}(x_1) \,\dl(x_1 - x_2)\, (b - 1)
+ \del_\rho W^\al(x_1) W^\7_\al(x_1) \,\dl(x_1 - x_2)\, (1 - a)
\notag \\
&\quad + [W \otto W^\7].
\label{eq:gauinv1} 
\end{align}
Generally, for two free fields $B(x)$, $C(x)$ it holds that the three
terms
$$
\del B(x_1)\,C(x_1) \,\dl(x_1 - x_2),\quad
B(x_1)\,\del C(x_1) \,\dl(x_1 - x_2) \word{and}
B(x_1)\,C(x_2) \,\del\dl(x_1 - x_2)
$$
are linearly independent. We conclude that condition
\eqref{eq:gauinv1} has a unique solution, namely $a = 1 = b$.
\textit{In fine}, the resulting contact term may be written as
\begin{equation}
\wt D_{\al\mu,\bt\rho}(y) - D_{\al\mu,\bt\rho}(y)
= \ihalf (2g_{\al\bt} g_{\mu\rho} - g_{\al\rho} g_{\mu\bt}
- g_{\mu\al} g_{\rho\bt}) \,\dl(y),
\label{eq:nisi-dominus-frustra} 
\end{equation}
which reproduces the seagull $e^2 AAWW^\7$-vertex in the Feynman rules
for the EW interaction.

\begin{figure}[htb]
\begin{tikzpicture}[>=Stealth, scale=1.3]
\begin{scope}[xshift=-3cm]
\coordinate (V3) at (0,0) ;   \coordinate (H) at (-1.3,0) ;
\coordinate (V1) at (1.7,1) ; \coordinate (V2) at (1.7,-1) ;
\coordinate (G1) at (3,1) ;   \coordinate (G2) at (3,-1) ;
\draw[dashed] (H) node[left=3pt] {$H$} 
   -- (V3) node[above=2pt] {$x_3$} ;
\draw[photon] (V1) -- (G1) node[right=3pt] {$\ga$} ;
\draw[photon] (V2) -- (G2) node[right=3pt] {$\ga$} ;
\draw[wplus] (V3) -- (V1) node[pos=0.5, below right=-4pt] {$W^+$} 
   node[above left] {$x_1$} ;
\draw[wplus] (V1) -- (V2) node[pos=0.5, left=-3pt] {$W^+$} 
   node[below left] {$x_2$} ;
\draw[wplus] (V2) -- (V3) node[pos=0.5, above right=-4pt] {$W^+$} ;
\foreach \pt in {V1,V2,V3} \fill (\pt) circle (1.6pt) ;
\end{scope}
\begin{scope}[xshift=3cm]
\coordinate (V3) at (0,0) ; \coordinate (H) at (-1.3,0) ;
\coordinate (V4) at (2,0) ; 
\coordinate (G1) at (3,1) ; \coordinate (G2) at (3,-1) ;
\draw[dashed] (H) node[left=3pt] {$H$} -- (V3) ;
\draw[photon] (V4) -- (G1) node[right=3pt] {$\ga$} ;
\draw[photon] (V4) -- (G2) node[right=3pt] {$\ga$} ;
\draw[wplus] (V3) arc[radius=1cm, start angle=180, end angle=0] 
   node[pos=0.5, below] {$W^+$} ;
\draw[wplus] (V4) arc[radius=1cm, start angle=0, end angle=-180] 
   node[pos=0.5, above] {$W^+$} ;
\foreach \pt in {V3,V4} \fill (\pt) circle (1.6pt) ;
\end{scope}
\end{tikzpicture}
\caption{Diagrams contributing to the amplitude to be computed}
\label{fg:diagrams} 
\end{figure}
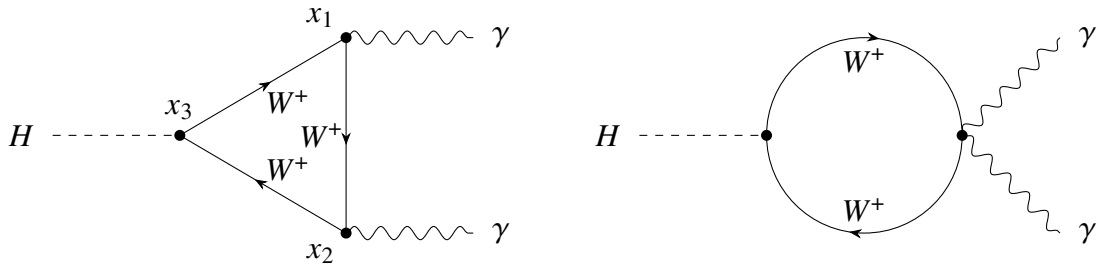

\marker 
\textit{Recapitulation}: assembling $T_3$ from $T_1$ and $T_2$ by the
Epstein--Glaser method, the above derived $AAWW^\7$-vertex (which is
part of $T_2$) generates the fish-like diagram in
Figure~\ref{fg:diagrams}. In the Feynman gauge, the $AAWW^\7$-vertex
was already derived in \cite{DuetschS99} by the same method. The fact
that we had to modify only $D^{\8\8}_{\8\8}$ (i.e., the
$G$-propagator), and not $\Dl^\8_\8$ (i.e., the $W$-propagator),
although here both are of the same singular order $\om = 0$, is in
accordance with the corresponding modification in the Feynman gauge
or, more generally, in a $R_\xi$ gauge: in such gauges the $AAWW^\7$
term \textit{can be derived in the same way}, but only the
$G$-propagator may be modified, because only this propagator has
$\om = 0$. The $WG$ and $GW$-propagators have singular order
$\om = -1$ (see Eqs.~\eqref{eq:extrema-exquisita}
and~\eqref{eq:remedia-optima-sunt}) and are not to be modified.

\marker
\textit{In conclusion}: almost the first thing to learn when working
in the BEG scheme is that what is usually taken from a top-down,
mysterious prevalence in particle physics of classical non-Abelian
gauge theories, with their quartic, second-order couplings, is here
\textit{derived} by pure quantum field theory operations. To wit, the
inductive construction of the time-ordered products respecting
``gauge invariance'', in the sense of \cite{DKS93, DHKS94,
DuetschS99, AsteDS99, Scharf01}, order by order in the couplings. In
fact, \textbf{all} of the reductive Lie algebra structure of the
Standard Model interactions, up to including at least one scalar
particle, comes \textit{automatically} in the BEG formalism from
respecting first principles of quantum field theory in building
the $\bS$-matrix%
\footnote{The derivation of the reductive Lie algebra structure was
annunciated by Stora~\cite{Stora-ESI}. His kind of principled
bottom-up approach has long been neglected in textbooks. A refreshing
exception is \cite[Problem~9.3]{MDSchwartz}.}
-- without invoking unobservable mechanisms. To go further into this
subject here would take us too far afield. We have merely dealt with
the example necessary for our purposes.

\subsection{Computation of the causal distribution at third order}
\label{ssc:eius-est-nolle}

We wish to mention here that triangular graphs in Epstein--Glaser
theory have been examined, for instance, in \cite[Chap.~3.8]{Scharf14}
-- the vertex correction in QED; in \cite{Axial} -- anomalies; as well
as in~\cite{BDF09} and \cite[Sect.~3.2]{Duetsch19}.%
\footnote{Some of these works are very instructive, in that they show
that cherished invariance properties cannot always be preserved
under distribution splitting.}
Of course, here we shall have more terms and a forest of indices.

Proceeding similarly to Sect.~\ref{sec:argumentum}, we perform the
adiabatic limit already at this stage. Since all the propagators in
the loop are massive, we may use \eqref{eq:dAAh}; hence, we only have
to compute the restricted distribution
$d^{\mu\nu}(k_1,k_2)|_{k_1^2=0=k_2^2}$. Again, because of the
kinematic constraints, only the first cut in formula
\eqref{eq:para-llorar} for $D_3$ counts (including the ``fish graph''
as derived in the previous subsection). Moreover, the discussions and
general conclusions of
Sects.~\ref{sec:quam-scriptum},~\ref{ssc:do-ut-des} and
\ref{ssc:hasta-ahi-podiamos-llegar} hold here, and will be assumed
without further ado.

\medskip

The two $WWA$-vertices are given in \eqref{eq:hic-iacet} and
\eqref{eq:lepus}, and the $WWh$-vertex reads
$$
T_1(x_3) = gM\, h(x_3) W^\nu(x_3) W^\7_\nu(x_3).
$$

The first step is to compute $R'^{A|B|C|D|F}_3(x_1, x_2; x_3)
:= T_2^{A|B|C|D|F}(x_1, x_2)\, T_1(x_3) \bigr|_\triangle$, where the
lower-order TOPs $T_2^\8$ have been respectively given by
lexicographical order in~\eqref{eq:TA(1,2)}, \eqref{eq:TB(1,2)},
\eqref{eq:TD(1,2)}, and finally \eqref{eq:sic-transit} together
with~\eqref{eq:nisi-dominus-frustra}. As it happens, their sum is
symmetric in~$(x_1,x_2)$, so we just introduce a general factor~$2$
and need not mention exchange anymore.

We start with
\begin{align}
R'^D_3(x_1,x_2;x_3) 
&= -2e^2gM\, h(x_3) F^{\mu\nu}(x_1) F^{\al\bt}(x_2)\,
r'^D_{\mu\nu,\al\bt}(y_1,y_2), \word{where}
\notag 
\\
r'^D_{\mu\nu,\al\bt}(y_1,y_2) 
&:= \Dl^{\ga\,+}_\nu(y_1) \,\Dl_{\mu\bt}(y_1 - y_2)
\,\Dl^+_{\al\ga}(y_2),  \word{and} y_k := x_k - x_3.
\label{eq:R'D_3-x} 
\end{align}
The ``fish'' diagram contribution reads:
\begin{align}
R'^F_3(x_1,x_2;x_3) &= -2e^2gM\, h(x_3) A^\mu(x_1) A^\nu(x_2)\,
r'^F_{\mu\nu}(y_1,y_2) \word{with}
\notag \\
r'^F_{\mu\nu}(y_1,y_2) &:= i\,\dl(y_1 - y_2)\, \bigl[ 
-g_{\mu\nu} \,\Dl^{\ga\,+}_\bt(y_1) \,\Dl^{\bt\,+}_\ga(y_1)
+ \Dl^{\ga\,+}_\nu(y_1) \,\Dl^+_{\mu\ga}(y_1) \bigr].
\label{eq:R'F_3-x} 
\end{align}
Next we compute $R'^B_3$:
\begin{align}
R'^B_3(x_1,x_2;x_3) &= -2e^2gM\, h(x_3) A^\mu(x_1) F_{\nu\bt}(x_2)\,
{r'^B}_\mu^{\nu\bt}(y_1,y_2)
\notag \\
\text{with}\quad {r'^B}_\mu^{\nu\bt}(y_1,y_2) 
&:= (g_\mu^\ga \del_\al - g^\ga_\al \del_\mu)\Dl^+(y_1)
\,\Dl^{\al\bt}(y_1 - y_2) \,\Dl^{\nu\,+}_\ga(y_2)
\notag \\
&\quad + \Dl^{\al\ga\,+}(y_1)\,
(g_\mu^\nu \del_\al - g^\nu_\al \del_\mu) \Delta^F(y_1 - y_2)
\,\Dl^{\bt\,+}_\ga(y_2).
\label{eq:R'B_3-x} 
\end{align}
The relation $R'^C_3(x_1,x_2;x_3) := R'^B_3(x_2,x_1;x_3)$ obviously
holds. Finally, for $R'^A_3$, we collect
\begin{align}
R'^A_3(x_1,x_2;x_3) 
&= -2e^2gM\, h(x_3) A^\mu(x_1) A^\nu(x_2)\, r'^A_{\mu\nu}(y_1,y_2),
\word{with}
\notag \\
r'^A_{\mu\nu}(y_1,y_2) 
&:= - (g_\mu^\ga \del_\al - g^\ga_\al \del_\mu)\Dl^+(y_1)
\,\Dl^{\al\bt}(y_1 - y_2)\,
(g_{\nu\ga} \del_\bt - g_{\bt\ga} \del_\nu)\Dl^+(y_2)
\notag \\
&\quad - \Dl^{\al\ga\,+}(y_1) \,D_{\al\mu,\bt\nu}(y_1 - y_2)
\,\Dl^{\bt\,+}_\ga(y_2)
\notag \\
&\quad - (g_\mu^\ga \del_\al - g^\ga_\al \del_\mu)\Dl^+(y_1)\,
(g^\al_\nu \del_\bt - g^\al_\bt \del_\nu)\Delta^F(y_1 - y_2)
\,\Dl^{\bt\,+}_\ga(y_2)
\notag \\
&\quad + \Dl^{\al\ga\,+}(y_1)\,
(g_{\mu\bt} \del_\al - g_{\al\bt} \del_\mu)\Delta^F(y_1 - y_2)\,
(g_{\ga\nu} \del^\bt - g^\bt_\ga \del_\nu)\Dl^+(y_2).
\label{eq:R'A_3-x} 
\end{align}

Next we express the resulting integrals by momentum space integrals.
By using $F^{\mu\nu}(k) = -i \bigl( k^\mu A^\nu(k) - k^\nu A^\mu(k)
\bigr)$ and $R'_3:=R'^A_3+R'^B_3+R'^C_3+R'^D_3+R'^F_3$ we gather
\begin{align}
&\int dx_1\,dx_2\,dx_3\, R'_3(x_1,x_2;x_3)
\label{eq:R'3-p} 
\\
&\quad = -2e^2gM (2\pi)^2\int dk_1\,dk_2\, h(k_1 + k_2)
A^\mu(-k_1) {A^\nu}(-k_2)\, r'_{\mu\nu}(k_1,k_2),
\notag
\\
&\word{where} r'_{\mu\nu}(k_1,k_2) := r'^A_{\mu\nu}(k_1,k_2) +
r'^F_{\mu\nu}(k_1,k_2)
\label{eq:r'(k)} 
\\
&+ \bigl[ ik_2^\bt (r'^B_{\mu,\bt\nu}(k_1,k_2) -
r'^B_{\mu,\nu\bt}(k_1,k_2)) + (k_1,\mu) \otto (k_2,\nu) \bigr]
\notag
\\
&- k_1^\bt k_2^\al \bigl[ r'^D_{\bt\mu,\al\nu}(k_1,k_2) -
r'^D_{\mu\bt,\al\nu}(k_1,k_2) - r'^D_{\bt\mu,\nu\al}(k_1,k_2) +
r'^D_{\mu\bt,\nu\al}(k_1,k_2) \bigr].
\notag
\end{align}

To compute the (combinations of) Fourier transformed
${r'}^{\dots}_{\dots}$-distributions appearing on the right hand side
of \eqref{eq:r'(k)}, we bring in $k_1^2 = 0 = k_2^2$, and omit all
terms containing a factor $k_{1\mu}$ or $k_{2\nu}$: this is allowed in
view of $k^\la A_\la(-k) = 0$. Similarly to the analogous computation
for the scalar model in the previous section, let us work with the
integration variable $q := k + k_2$, and introduce $P := k_1 + k_2$,
hence $k_1 - k = P - q$ and $s := P^2 = 2k_1k_2$ and $Pk_1 = k_1k_2 =
Pk_2$. Due to the factors $\Dl^+(q)$ and $\Dl^+(P - q)$, we may use
the relations
\begin{gather*}
q^2 = M^2, \quad (P - q)^2 = M^2, 
\word{hence}  2Pq = s = 2(k_1k_2),  \word{implying}
\\
(P - q)q = (k_1k_2) - M^2, \quad 0 = (P(q - k_2)) = (Pk), \quad
(q - k_2)^2 = M^2 - 2(qk_2),
\end{gather*}
and $((P - q)\,k_1) = (q k_2)$, that is, $(qk_1) + (qk_2) = (k_1k_2)$.
Remember also that we may replace $P_\mu \to k_{2\mu}$ and
$P_\nu \to k_{1\nu}$.

To tally the fish diagram contribution, we introduce the integrals
\begin{equation}
J^{\{\.|\mu|\mu\nu\}}(P) 
:= \int d^4q\, \{1|q^\mu|q^\mu q^\nu\} \,\Dl^+(P - q)\,\Dl^+(q).
\label{eq:J-def} 
\end{equation}
These symbols $J^\8(P)$ generalize that of the scalar $J(P)$, defined
in~\eqref{eq:def-J}. We obtain:
\begin{align}
r'^F_{\mu\nu}(k_1,k_2) &= \frac{i}{(2\pi)^4} \biggl( g_{\mu\nu} \Bigl(
-2 + \frac{s}{M^2} - \frac{s^2}{4M^4} \Bigr) J(P) - \frac{k_{2\mu}
k_{1\nu}}{M^2} J(P)
\notag
\\
&\hspace{3em} + \frac{1}{M^2} \Bigl( k_{2\mu} J_\nu(P) + k_{1\nu}
J_\mu(P)\,\frac{s}{2M^2} \Bigr) - J_{\mu\nu}(P)\Bigl( \frac{1}{M^2} +
\frac{s}{2M^4} \Bigr) \biggr).
\label{eq:r'F-p} 
\end{align}
 
To figure out $J^\mu$ and $J^{\mu\nu}$ one first argues that
\begin{align}
J^\mu(P) &= \th(P_0)\,\th(s - 4M^2)\,P^\mu\,g(s),
\notag \\
J^{\mu\nu}(P)
&= \th(P_0)\,\th(s - 4M^2)\, (P^\mu P^\nu\,a(s) + g^{\mu\nu} c(s)),
\label{eq:J-1} 
\end{align}
for appropriate $g(s), \,a(s)$ and $c(s)$. The latter can be obtained
from the identites
$$
J^\mu(P) P_\mu = \half sJ(P), \quad
J^\mu_\mu(P) = M^2 J(P), \quad
J^{\mu\nu}(P) P_\mu P_\nu = \quarter s^2J(P).
$$ 
and from $J(P)=\th(P^0)\, \th(s - 4M^2)\,4i\pi^2\,s\, G(s)$ -- see
\eqref{eq:G}. So we arrive at
$$
g(s) = 2i\pi^2 s\,G(s); \quad
a(s) = \frac{4i\pi^2}{3}\,(s - M^2) G(s); \quad
c(s) = \frac{4i\pi^2}{3}\, (M^2 - s/4) s\,G(s),
$$
yielding
\begin{align}
r'^F_{\mu\nu}(k_1,k_2) &= \frac{1}{4\pi^2}\, \th(P^0)\, \th(s -
4M^2)\, G(s) \biggl[ \Bigl( \frac{14}{3} - \frac{11}{6}\,\frac{s}{M^2}
+ \frac{5}{12}\,\frac{s^2}{M^4} \Bigr) g_{\mu\nu}(k_1 k_2)
\notag
\\
&\quad + k_{1\nu} k_{2\mu} \Bigl( -\frac{1}{3} +
\frac{2}{3}\,\frac{s}{M^2} - \frac{1}{12}\,\frac{s^2}{M^4} \Bigr)
\biggr].
\label{eq:r'F-final} 
\end{align}

\paragraph{Other new integrals.}
We introduce already all the new integrals required in this section.
By means of $(q - k_2)^2 - M^2 = - 2(qk_2)$, implying
$(qk_2)\,\Delta^F(q - k_2) = - i/8\pi^2$, some of the integrals to
appear are calculated:
\begin{align*}
K^{\{\.|\mu|\mu\nu\}}(P)& := \int \! d^4q\, \{1|q^\mu|q^\mu q^\nu\}\,
(qk_2)\,\Dl^+(P - q) \Delta^F(q - k_2)\,\Dl^+(q)
= -\frac{i}{8\pi^2}\, J^{\{\.|\mu|\mu\nu\}}(P),\\
L(P) &:= \int \! d^4q\, (qk_2)^2 \,\Dl^+(P - q) \Delta^F(q - k_2)\,\Dl^+(q)
= k_{2\mu}\,K^\mu .
\end{align*}

The following integrals will also be needed:
\begin{align*}
N^{\{\.|\mu|\mu\nu\}}(P) &:= \int d^4q\, 
\{1|q^\mu|q^\mu q^\nu\} \,\Dl^+(P - q) \Delta^F(q - k_2) \,\Dl^+(q).
\end{align*}
By using $k_2^\nu=0$, they are easily be expressed in terms of the
integrals computed in Sect.~\ref{ssc:d-scalar}:
$$
N(P) = I(P)\,,\qquad 
N^\mu(P) = I^\mu(P) + k_2^\mu I(P)\,,\qquad
N^{\mu\nu}(P) = I^{\mu\nu}(P) + k_2^\mu I^\nu(P)\,, 
$$

\paragraph{Electromagnetic gauge invariance.}
The computation of the individual terms in \eqref{eq:r'(k)} is
lengthy, but straightforward. For the $r'^B$-terms we reap
\begin{align*}
& ik_2^\bt \bigl( 
r'^B_{\mu,\bt\nu}(k_1,k_2) - r'^B_{\mu,\nu\bt}(k_1,k_2) \bigr) 
= \frac{1}{(2\pi)^2} \biggl\{ 
- 2N(P)(g_{\mu\nu} (k_1k_2) - k_{1\nu} k_{2\mu})
\notag\\
&\quad + g_{\mu\nu} (k_1k_2) \biggl(
2\frac{K(P)}{M^2} - \frac{L(P)}{M^4} \biggr)
- k_{1\nu} k_{2\mu} (k_1k_2)\,\frac{K(P)}{M^4}
+ N_{\mu\nu}(P) \biggl( 
2\frac{(k_1k_2)}{M^2} - \frac{(k_1k_2)^2}{M^4} \biggr)
\notag\\
&\quad + k_{2\mu} \biggl( N_\nu(P) \Bigl( - 2\frac{(k_1k_2)}{M^2}
+ \frac{(k_1k_2)^2}{M^4}\, \Bigr) 
+ K_\nu(P)\,\frac{(k_1k_2)}{M^4} \biggr)
+ k_{1\nu}\,K_\mu(P) \biggl( - \frac{2}{M^2} 
+ \frac{k_1k_2}{M^4} \biggr) \biggr\}.
\end{align*}
Inserting the integrals calculated above, one reaches for $ik_2^\bt
\bigl( r'^B_{\mu,\bt\nu}(k_1,k_2) - r'^B_{\mu,\nu\bt}(k_1,k_2) \bigr)
+ (k_1,\mu) \otto (k_2,\nu)$ the following result:
\begin{align}
\frac{\th(P^0)\,\th(s - 4M^2)}{(2\pi)^2}\,
\bigl( g_{\mu\nu} (k_1k_2) - k_{1\nu} k_{2\mu} \bigr)\, \biggl(
F(s) \Bigl( -2 - \frac{s}{2M^2} \Bigr) + G(s)\,\frac{s}{M^2} \biggr).
\label{eq:rB-final} 
\end{align}
Note that $ik_2^\bt \bigl( r'^B_{\mu,\bt\nu}(k_1,k_2)
- r'^B_{\mu,\nu\bt}(k_1,k_2) \bigr)$ is individually gauge invariant.
This reflects the known fact that $T_2^B$ given by~\eqref{eq:TB(1,2)}
is individually gauge invariant.

Analogously, since $T_2^D$ \eqref{eq:TD(1,2)} is trivially gauge
invariant, we expect the combination
$$
- k_1^\bt k_2^\al \bigl( r'^D_{\bt\mu,\al\nu}(k_1,k_2)
- r'^D_{\mu\bt,\al\nu}(k_1,k_2) - r'^D_{\bt\mu,\nu\al}(k_1,k_2)
+ r'^D_{\mu\bt,\nu\al}(k_1,k_2) \bigr),
$$
to contain the factor $g_{\mu\nu} (k_1k_2) - k_{1\nu} k_{2\mu}$.
Indeed, we obtain
\begin{align}
& \frac{1}{(2\pi)^2} \biggl\{
-2(g_{\mu\nu} (k_1k_2) - k_{1\nu} k_{2\mu})\, N(P)
+ g_{\mu\nu} \biggl( -2 \frac{L(P)}{M^2} 
+ (k_1k_2) \Bigl( 2\frac{K(P)}{M^2} - \frac{L(P)}{M^4} \Bigr) \biggr)
\notag \\
&\quad + k_{1\nu} k_{2\mu}\, (k_1k_2) \,\frac{K(P)}{M^4}
+ N_{\mu\nu} \biggl( 2\frac{(k_1k_2)}{M^2} 
+ \frac{(k_1k_2)^2}{M^4} \biggr) 
+ k_{1\nu} \biggl( -2 \frac{K_{1\mu}(P)}{M^2} 
- (k_1k_2)\,\frac{K_{\mu}(P)}{M^4} \biggr)
\notag
\\
&\quad + k_{2\mu} \biggl( 2\frac{K_{\nu}(P)}{M^2}
- 2\frac{(k_1k_2)}{M^2}\,N_\nu(P) - \frac{(k_1k_2)^2}{M^4}\,N_\nu(P)
+ (k_1k_2)\,\frac{K_{1\nu}(P)}{M^4} \biggr) \biggr\}
\notag \\
&= \frac{\th(P^0)\,\th(s - 4M^2)}{(2\pi)^2}\,
(g_{\mu\nu} (k_1k_2) - k_{1\nu} k_{2\mu}) \biggl[
\Bigl(-1 + \frac{s}{4M^2} \Bigr)F(s) - \frac{s^2}{4M^4}\,G(s) \biggr].
\label{eq:rD-final} 
\end{align}

Finally, for~$r'^A$ we get
\begin{align*}
r'^A_{\mu\nu}(k_1,k_2)
&= \frac{1}{(2\pi)^2} \biggl\{
-2(g_{\mu\nu} (k_1k_2) - k_{1\nu} k_{2\mu})\,N(P)
\notag \\
&\quad + g_{\mu\nu} \biggl( 
\Bigl( 2 + \frac{2(k_1k_2)}{M^2} \Bigr)\,K(P)
- \Bigl( \frac{2}{M^2} + \frac{(k_1k_2)}{M^4} \Bigr)\,L(P) \biggr)
\notag \\
&\quad + k_{1\nu} k_{2\mu} \Bigl( -\frac{(k_1k_2)}{M^4} \Bigr)\,K(P)
+ k_{1\nu} \Bigl( -\frac{2}{M^2} + \frac{(k_1k_2)}{M^4} \Bigr)\,K_\mu(P)
\notag \\
&\quad + k_{2\mu} \biggl( \Bigl( -12 + 2\frac{(k_1k_2)}{M^2}
- \frac{(k_1k_2)^2}{M^4} \Bigr)\,N_\nu(P)
+ \Bigl( \frac{4}{M^2} + \frac{(k_1k_2)}{M^4} \Bigr)\,K_\nu(P) \biggr)
\notag \\
&\quad + \Bigl( 12 - 2\frac{k_1k_2}{M^2} 
+ \frac{(k_1k_2)^2}{M^4} \Bigr)\,N_{\mu\nu}(P)
- \Bigl( \frac{2}{M^2} + 2\frac{(k_1k_2)}{M^4} \Bigr)\,K_{\mu\nu}(P)
\biggr\}.
\end{align*}
After insertion of the integrals, this is equal to
\begin{align*}
&\frac{\th(P^0)\,\th(s - 4M^2)}{(2\pi)^2} \biggl\{
(g_{\mu\nu} (k_1k_2) - k_{1\nu} k_{2\mu})\,F(s)\,
\Bigl(-3 + 12\,\frac{M^2}{s} + \frac{s}{4M^2} \Bigr)
\\
&+ G(s) \biggl[ g_{\mu\nu} (k_1k_2) \Bigl( -\frac{14}{3}
+ \frac{5}{6}\,\frac{s}{M^2} - \frac{1}{6}\,\frac{s^2}{M^4} \Bigr)
+ k_{1\nu} k_{2\mu} \Bigl( \frac{1}{3} + \frac{s}{3M^2} 
- \frac{s^2}{6M^4} \Bigr) \biggr] \biggr\}.
\end{align*}
The sum $r'^A + r'^F$ is indeed gauge invariant:
\begin{align*}
r'^A_{\mu\nu}(k_1,k_2) + r'^F_{\mu\nu}(k_1,k_2)
&= \frac{\th(P^0)\,\th(s - 4M^2)}{(2\pi)^2}\,
(g_{\mu\nu} (k_1k_2) - k_{1\nu} k_{2\mu})
\\
&\quad \x \biggl( F(s)\,\Bigl( -3 + 12\,\frac{M^2}{s} 
+ \frac{s}{4M^2} \Bigr) + G(s)\,\Bigl( - \frac{s}{M^2}
+ \frac{s^2}{4M^4} \Bigr) \biggr),
\end{align*}
as expected from the outcome of the discussion in subsection
\ref{ssc:per-angusta}.

The $G(s)$-terms are seen to cancel out, and for the total $r'$ we
obtain the following gauge-invariant result:
\begin{equation}
r'_{\mu\nu}(k_1,k_2) = \frac{\th(P^0)\,\th(s - 4M^2)}{(2\pi)^2}\,
(g_{\mu\nu} (k_1k_2) - k_{1\nu} k_{2\mu})\,6 F(s)\,(-1 + 2M^2/s).
\label{eq:r'-final} 
\end{equation}
Like for the scalar model, $a'_{\mu\nu}(k_1,k_2)$ is obtained from
$r'_{\mu\nu}(k_1,k_2)$ in~\eqref{eq:r'-final} by replacing $\th(P^0)$
by $\th(-P^0)$, and hence $d_{\mu\nu} = a'_{\mu\nu} - r'_{\mu\nu}$ by
replacing $\th(P^0)$ by~$-\sgn(P^0)$.

Hence $d^{\mu\nu}_\gi(k_1, k_2)|_{k_1^2=0=k_2^2}$ is of the form
\begin{align}
&d^{\mu\nu}_\gi(k_1,k_2)\bigr|_{k_1^2=0=k_2^2}
= P^{\mu\nu}(k_1,k_2)\,d_0(P), \word{where}
\label{eq:d-onshell-a} 
\\
&d_0(P) := \frac{i\,\sgn(P^0)\,\th(\tilde\rho - 1)}{(2\pi)^5}\,
b_1(\tilde\rho), \quad
b_1(\tilde\rho) := -\frac{3}{16 M^2\,\tilde\rho}\,
\biggl( \frac{1}{\tilde\rho} - 2 \biggr)\,
\log\frac{1 + \sqrt{1 - \tilde\rho^{-1}}}
{1 - \sqrt{1 - \tilde\rho^{-1}}}\,.
\notag
\end{align}
The factor $3$ above was to be expected, since the Proca field has
three components. The singular order of the on-shell distribution
$d^{\mu\nu}_\gi(k_1,k_2)|_{k_1^2=0=k_2^2}$ is clearly equal to zero.
It is also easy to show that power counting rules imply that the
singular order of the \textit{off-shell} $d$-distribution%
\footnote{This notation is badly abused in this paper. But that is
hardly avoidable.}
$d^{\mu\nu}$ satisfies the bounds 
$6 \geq \om := \om(d^{\mu\nu}) \geq 2$.

\subsection{Distribution splitting in EW theory}
\label{ssc:quod-feceris}

The off-shell $d^{\mu\nu}$-distribution for the diphoton decay of the
higgs via EW vector bosons is also of the genre~\eqref{eq:dV1}, the
difference being that the pertinent polynomials $h_j^{\mu\nu}$ are of
a higher degree, i.e., $\om := \om(d^{\mu\nu})$ has a greater value.
Therefore, the distribution splitting method for null momenta
developed in Section~\ref{sec:quam-scriptum} does apply:
$t^{\mu\nu}_\gi|_\sM = a^{c\,\mu\nu}_\gi|_\sM$ can be computed by
inserting the light-cone restriction of $d^{\mu\nu}_\gi(k_1,k_2)$ into
the dispersion integral~\eqref{eq:dispersion_simplified}. Again, this
restricted $d$-distribution is of the genre~\eqref{eq:d0-h-h1} for the
given~$b_1$. Hence, the \textit{central solution} compatible with EGI
is obtained by:
\begin{align}
& t^{c\,\mu\nu}_\gi(k_1,k_2) = P^{\mu\nu}(k_1,k_2)\,t^c_\gi(P),
\word{where}
\notag \\
& t^c_\gi(P) = \frac{i\eta}{2\pi} \int_{|t|\geq\sqrt{1/\tilde\rho}}\,
\frac{t^2\,d_0(tP)\,dt}{(1 - t)\,t^{\max\{\om+1,\,0\}}} 
= -\frac{1}{(2\pi)^6} \int_{1/\tilde\rho}^\infty
\frac{u\,b_1(u\tilde\rho)\,du} {u^{\max\{\piso{\om/2}+1,\,0\}}(1-u)},
\label{eq:t-dispint} 
\end{align}
for $(k_1,k_2) \in \sM$. Let us first assume that $\om =2$. Before
computing, we recall that the \textit{ambiguity} of the result for
$t^{\mu\nu}_\gi$ will be given in that case by a polynomial of
\textit{second} degree in~$(k_1,k_2)$ containing the factor
$P^{\mu\nu}$; that is, in the occasion, a constant multiple
of~$P^{\mu\nu}$ itself.

Besides a global factor $-3(16M^2 (2\pi)^6)^{-1}$, in view of
\eqref{eq:d-onshell-a} and making the customary change of variable
$v=u\tilde\rho$, to obtain $t^c_\gi(\tilde\rho)$ we ought to compute:
\begin{align*}
&\tilde\rho\int_1^\infty\!dv\,\Biggl[ 
\frac{\log\bigl( 1 + \sqrt{1 - v^{-1}} \bigr) 
/ \bigl( 1 - \sqrt{1 - v^{-1}} \bigr)}{v^3(\tilde\rho - v)}
- \frac{2 \log\bigl( 1 + \sqrt{1 - v^{-1}} \bigr)
/ \bigl( 1 - \sqrt{1 - v^{-1}} \bigr)}{v^2(\tilde\rho - v)} \Biggr]
\\
&= \int_1^\infty\! dv\, \Biggl[ 
\frac{\log\bigl( 1 + \sqrt{1 - v^{-1}} \bigr) 
/ \bigl( 1 - \sqrt{1 - v^{-1}} \bigr)}{v^2(\tilde\rho - v)}
- \frac{2 \log\bigl( 1 + \sqrt{1 - v^{-1}} \bigr) 
/ \bigl( 1 - \sqrt{1 - v^{-1}} \bigr)}{v(\tilde\rho - v)} \Biggr]
\\
&+ \int_1^\infty\! dv\, \Biggl[ 
\frac{\log\bigl( 1 + \sqrt{1 - v^{-1}} \bigr) 
/ \bigl( 1 - \sqrt{1 - v^{-1}} \bigr)}{v^3}
- \frac{2 \log\bigl( 1 + \sqrt{1 - v^{-1}} \bigr) 
/ \bigl( 1 - \sqrt{1 - v^{-1}} \bigr)}{v^2} \Biggr].
\end{align*}
Notice that the last two integral terms just yield a number, equal to
minus the sum of the values of the two previous ones at
$\tilde\rho = 0$. The first integral term is already known from the
discussion in subsection~\ref{ssc:splitting}, and in all yields
$$
\frac{3}{8M^2(2\pi)^6} J_2(\tilde\rho)  
= \frac{-3}{8M^2(2\pi)^6} \biggl(
\frac{1}{\tilde\rho} - \frac{f(\tilde\rho)}{\tilde\rho^2} \biggr).
$$
This is essentially $3F_0(\tilde\rho)$, that clearly contributes to
the expected final result for $F_1(\tilde\rho)$ -- recall
\eqref{eq:magister-dixit}. According to Eq.~\eqref{eq:J1} -- also in
subsection~\ref{ssc:splitting} -- in the case $a=0$, the second
integral term yields
$$
- \frac{2\cdot 3}{8M^2(2\pi)^6}\,J_1(\tilde\rho,0)= 
- \frac{6}{8M^2(2\pi)^6}\, \frac{f(\tilde\rho)}{\tilde\rho}\,.
$$
Including the last two integrals, we recover
$t^{c}_{\gi}(\tilde\rho)\sim F_1(\tilde\rho) - 7$. Taking into account
the Epstein--Glaser ambiguity as limited by EGI, and bringing in
constant prefactor appearing in~\eqref{eq:R'3-p}, we end up with the
amplitude
\begin{equation}
- 2e^2gM (2\pi)^2\,t^{\mu\nu}_{\gi}(k_1,k_2) 
= \frac{e^2g}{4(2\pi)^4\,M}\, P^{\mu\nu}(k_1,k_2)\,
\bigl( F_1(\tilde\rho) + C \bigr),
\label{eq:latapadelperol} 
\end{equation}
with $C$ an arbitrary constant. 

If $\om > 2$, the central solution within the unitary gauge is
obtained by adding the following expression to the above obtained
result, more precisely, to $-(8M^2(2\pi)^6)^{-1} \bigl(
F_1(\tilde\rho) - 7 \bigr)$:
\begin{align}
& -\frac{1}{(2\pi)^6} \int_{\tilde\rho^{-1}}^\infty du\,
\frac{b_1(u\tilde\rho)}{1 - u}
\biggl( \frac{1}{u^{\piso{\om/2}}} - \frac{1}{u} \biggr)
= -\frac{1}{(2\pi)^6} \int_{\tilde\rho^{-1}}^\infty du\,
b_1(u\tilde\rho) \sum_{k=0}^{\piso{\om/2}-2}
\frac{1}{u^{\piso{\om/2} - k}}
\nonumber
\\
&\quad = \sum_{k=0}^{\piso{\om/2}-2} \tilde c_k\,
\tilde\rho^{\piso{\om/2}-1-k}; \word{where the} 
\tilde c_k := -\frac{1}{2\pi} \int_1^\infty dv\,
\frac{b_1(v)}{v^{\piso{\om/2}-k}}
\label{eq:t-t0} 
\end{align}
do not depend on $\tilde\rho$; in the first step we have used the
relation
$$
\frac{1}{u^{\piso{\om/2}}} - \frac{1}{u}
= (1 - u) \sum_{k=0}^{\piso{\om/2}-2} \frac{1}{u^{\piso{\om/2}-k}}\,,
$$
and in the last step we have made again the substitution $v :=
u\tilde\rho$. The point is that \eqref{eq:t-t0} is a polynomial in
$\tilde\rho$ allowed by the surviving Epstein--Glaser ambiguity.
Therefore, without knowing the precise value of $\om$, the general
Epstein--Glaser solution respecting EGI can be written as
\begin{equation}
- 2e^2gM (2\pi)^2\,t^{\mu\nu}_\gi(k_1,k_2)
= \frac{e^2g}{4(2\pi)^4\,M}\, P^{\mu\nu}(k_1,k_2)\, \biggl(
F_1(\tilde\rho) + \sum_{k=0}^{\piso{\om/2}-1} C_k\tilde\rho^k \biggr),
\label{eq:gratatio-capitis} 
\end{equation}
where the constants $C_k$ are all arbitrary. However, terms
corresponding to $\om\geq 4$ can be discarded on grounds of
perturbative unitarity.

\medskip

Let us now to come back to reference~\cite{ChristovaI}. It is
argued there that the convergent integral
$$
t_0(\tilde\rho) := \int_{\tilde\rho^{-1}}^\infty du\,
\frac{b_1(u\tilde\rho)}{1 - u}
$$
leads to the correct result. From the standpoint of this reference,
formula~\eqref{eq:t-dispint} is ``oversubtracted''. One obtains there,
yet again
\begin{align}
P^{\mu\nu} t_0(\tilde\rho) 
&= P^{\mu\nu} \! \int_{\tilde\rho^{-1}}^\infty du\, 
\frac{b_1(u\tilde\rho)}{1 - u} 
= - \frac{3\,P^{\mu\nu}}{8M^2(2\pi)^6}\,
(2 J_1(\tilde\rho,0) - J_2(\tilde\rho))
\nonumber \\
&= - \frac{P^{\mu\nu}}{8M^2(2\pi)^6} \,\biggl[ 
\frac{3}{\tilde\rho} + \frac{6f(\tilde\rho)}{\tilde\rho}
- \frac{3f(\tilde\rho)}{\tilde\rho^2} \biggr]
= - \frac{P^{\mu\nu}}{8M^2(2\pi)^6}\,(F_1(\tilde\rho) - 2).
\label{eq:disp-int-solu} 
\end{align}
So the naive on-shell computation yields a \textit{particular}
Epstein--Glaser solution. In the present case, however,
equation~\eqref{eq:gratatio-capitis} tells us that we are
\textit{forced} to add (at least) a polynomial of degree two
respecting EGI, that is, a term $C P^{\mu\nu}$ for~$C$ an
indeterminate constant -- with which our result for the amplitude is
compatible with the generally accepted one. 

\begin{remk} 
From our viewpoint, the expression in \eqref{eq:disp-int-solu} is the
unique Epstein--Glaser solution respecting EGI, corresponding to the
following causal $d$-distribution: let the result~\eqref{eq:d-onshell-a} for
$d^{\mu\nu}_\gi(k_1,k_2)$ (obtained by light-cone restriction of the
photon momenta) be \textit{interpreted as an unrestricted element of}
$\sS'(\bR^8)$, that is, all values $(k_1,k_2) \in \bR^8$ are admitted.
One easily verifies that this $d$-distribution has causal support, so
the splitting problem is well defined, and since its singular order is
zero, the EGI requirement selects a unique splitting solution. Writing
the latter suitably as a dispersion integral in momentum space, one
verifies the claim. This procedure strongly simplifies explicit
computations, but it is not conceptually correct.%
\footnote{Actually, in the first edition of the book by Scharf on
quantum electrodynamics (i.e.,~\cite{Scharf89} rather than
\cite{Scharf14}), the vertex function in QED at third order was
computed by such a method.}
\end{remk}

\subsection{Fixing the normalization polynomial by agreement with the
Feynman gauge}
\label{ssc:ita}

In order to determine the normalization polynomial we may as well
invoke the computation of the $h \to \ga\ga$ decay in the Feynman
gauge and \textit{gauge-fixing independence}, namely, the requirement
that observable quantities should not depend on the choice of gauge.%
\footnote{The equivalence or inequivalence of calculations performed
in different gauges was a nagging worry of Raymond Stora in his last
years. The classic paper~\cite{ChristTDLee} illustrates the 
difficulties lurking here.}
Motivated by results of \cite{AsteSD98},%
\footnote{This reference works with a formulation of gauge invariance
suitable for the BEG scheme. In that framework it was shown for
the various $R_\xi$-gauges that the $T$-products can be normalized in
such a way that the physical $S$-matrix (i.e., for in- and out-states
being on-shell) does not depend on the gauge-fixing parameter~$\xi$ in
the formal adiabatic limit; and that this normalization is compatible with
gauge invariance in the mentioned sense.}
we contend that the ``entirely on-shell'' amplitude coming out of our
previous computation should coincide with that of an Epstein--Glaser
computation in the Feynman gauge. By ``entirely on-shell'' we mean
that not only the photons, but \textit{also the higgs} is on-shell,
that is, $\tilde\rho = \rho := m_h^2/4M^2$.

Denote the Epstein--Glaser result for the $d$-distribution in the
Feynman gauge by $d^1_{\mu\nu}$. In contrast with the unitary gauge,
there additionally contribute diagrams with St\"uckelberg fields and
Faddeev--Popov ghosts (as inner lines) to $d^1_{\mu\nu}$, see
e.g.~\cite{WusStrikeBack}. We spare the reader the details of the
construction of the TOPs, and in particular the derivation of the
$AAWW^\7$-vertex in this context. For photons on-shell with physical
polarizations (setting $k_1^2 = 0 = k_2^2$ and omitting pure gauge
terms $\sim k_{1\mu}$ or $\sim k_{2\nu}$), our result reads:
\begin{align}
&d_{\mu\nu}^1(k_1,k_2) = -\frac{1}{2^3\,(2\pi)^6\,M^2}
\label{eq:absorptive_r_xi}\x  
\\
&\bigg[ (k_1 k_2)g_{\mu\nu} \left(-\frac{3}{\tilde\rho^2}
+\frac{7}{\tilde\rho} -\frac{\rho}{\tilde\rho^2} \right) -k_{2\mu}
k_{1\nu} \left(-\frac{3}{\tilde\rho^2}+\frac{8}{\tilde\rho} -
\frac{2\rho}{\tilde\rho^2} \right) \bigg]\Im f(\tilde\rho).
\notag
\end{align}

The tedious computation of the above absorptive part was done with the
aid of the Mathematica package FeynCalc~\cite{feyncalc}. The
computation proceeds along the lines of the computations of related
absorptive parts in scalar electrodynamics and electroweak theory in
the unitary gauge presented in full detail in Sec.~\ref{ssc:d-scalar}
and \ref{ssc:eius-est-nolle}, respectively. As before, all terms
contributing to the distribution $d_{\mu\nu}^1$ can be represented by
Feynman diagrams with cuts -- for the complete list see e.g.
\cite{WusStrikeBack}. Just like in subsections.~\ref{ssc:d-scalar}
and~\ref{ssc:eius-est-nolle}, because of the kinematic constraints one
needs to consider only the cut separating the higgs vertex from the
photon vertices. All the appearing expressions have a very similar
structure to those that have been already considered in the
above-mentioned parts. Thanks to the presence of the cut, each
integral over the four-momentum flowing in the loop can be converted
into an integral over a sphere, which can be evaluated explicitly. We
stress the fact that, due to compactness of the region of integration,
the computation of the absorptive part does not involve any
regularization.

An important feature of electroweak theory in the $R_\xi$-gauges is
the fact that all interaction vertices have dimensions lower or equal
to four (because $\mathrm{dim}\,W^\mu=1$, in contrast to the value
$\mathrm{dim}\,W^\mu=2$ for the unitary gauge). In particular, a
straightforward power counting argument gives the upper bound
$\om\bigl(d^1_{\mu\nu}\bigr)\leq0$ for the singular order of the
off-shell distribution $d_{\mu\nu}^1$. Noting that $\tilde\rho=(k_1
k_2)/2M^2$ and $f(\tilde\rho)= O(\log\tilde\rho)$ we see that the
on-shell restriction of $d^1_{\mu\nu}(k_1,k_2)$ in
Eq.~\eqref{eq:absorptive_r_xi} grows logarithmically for big values of
$\tilde\rho$. For the off-shell~$d_{\mu\nu}^1$, this implies the
equality $\om\bigl(d^1_{\mu\nu} \bigr)=0$. This should be contrasted
with the bounds $6\geq\om\bigl(d_{\mu\nu} \bigr) \geq 2$ in the case
of the absorptive part computed in the unitary gauge.

The off-shell distribution $d_{\mu\nu}^1$ is again of the type
considered in Sec.~\ref{ssc:do-ut-des}. In particular, the method of
distribution splitting developed in subsection
\ref{ssc:hasta-ahi-podiamos-llegar} is applicable. For photons
on-shell with physical polarizations, the central solution reads
\begin{align}
&t^{1\,c}_{\mu\nu}(k_1,k_2) 
= -\frac{1}{2^3\,(2\pi)^6\,M^2}\,\x
\label{eq:txi-on} 
\\
&\biggl\{g_{\mu\nu} (k_1k_2) \biggl[ \Bigl(
-\frac{3}{\tilde\rho^2} + \frac{7}{\tilde\rho} -
\frac{\rho}{\tilde\rho^2}\Bigr) f(\tilde\rho) + \frac{3}{\tilde\rho} +
\frac{2\rho}{\tilde\rho} \biggr] - k_{1\nu} k_{2\mu} \biggl[ \Bigl(
-\frac{3}{\tilde\rho^2} + \frac{8}{\tilde\rho} -
\frac{2\rho}{\tilde\rho^2} \Bigr) f(\tilde\rho) + \frac{3}{\tilde\rho}
+ \frac{2\rho}{\tilde\rho} \biggr] \biggr\}.
\notag
\end{align} 
According to the postulate `Divergence degree' (in Sect.
\ref{ssc:porca-miseria}), we have to demand for the off-shell
$t^1_{\mu\nu}$ that
\begin{equation*}
\om\bigl(t^1_{\mu\nu} \bigr) = \om\bigl(d^1_{\mu\nu}\bigr) = 0.
\end{equation*}
This implies that the pertaining normalization freedom consists of a
constant term which is a tensor with two indices. By the Lorentz
invariance such a term has to be proportional to the metric tensor.
Consequently, the general off-shell solution of the splitting problem
is of the form
\begin{equation}
t^{1}_{\mu\nu}(k_1,k_2) = t^{1\,c}_{\mu\nu}(k_1,k_2) + g_{\mu\nu} D,
\end{equation} 
where $D$ is an arbitrary constant; note that this relation holds also
after restriction to on-shell photons with physical polarizations.

Observe that, in contrast to the unitary gauge, as long as the higgs
is off-shell, the distributions~\eqref{eq:absorptive_r_xi}
and~\eqref{eq:txi-on} are \textit{not} electromagnetically
gauge-invariant. This was to be expected and is related to the
presence of unphysical degrees of freedom in electroweak theory in the
$R_\xi$-gauges. However, entirely on-shell EGI can be satisfied:
setting $\tilde\rho := \rho$ in~\eqref{eq:txi-on}, we plainly get
\begin{equation}
t^1_{\mu\nu}(k_1,k_2)\bigr|_{\tilde\rho=\rho}
= - \frac{1}{2^3\,(2\pi)^6\,M^2}\, \bigl( 
P_{\mu\nu}(k_1,k_2)\, F_1(\rho) + D\,g_{\mu\nu} \bigr),
\label{eq:txi-onon} 
\end{equation}
and one sees that $D$ must be put equal to zero. This fixes completely
the normalization freedom in the construction of~$t_{\mu\nu}^1$ in the
Feynman gauge. At this level there is of course coincidence with the
result in \cite{HerreroMorales}, despite different game rules.

Recall that in the unitary gauge, for on-shell photons with physical
polarizations, the general normalization freedom fulfilling
electromagnetic gauge invariance and Lorentz covariance is given by
the last term in \eqref{eq:gratatio-capitis}, where $\om\equiv
\om(d_{\mu\nu})$. We stress that the constants~$C_k$ appearing in that
term cannot be fixed without imposing some further normalization
conditions. To address this problem, observe
that it is possible to adjust the coefficients $C_k$ of the polynomial
in the expression~\eqref{eq:gratatio-capitis} for $t_{\gi,\mu\nu}$ in
the unitary gauge in such a way that the following equality
\begin{equation}
t_{\gi,\mu\nu}(k_1,k_2)\bigr|_{\tilde\rho=\rho}
= t^1_{\mu\nu}(k_1,k_2)\bigr|_{\tilde\rho=\rho}.
\label{eq:gauindep} 
\end{equation}
holds entirely on-shell, i.e. for $\tilde\rho=\rho$. In fact, we must
set $C_0 := 2$ and $C_k := 0$ for all $k \geq 1$
in~\eqref{eq:gratatio-capitis}, which fixes completely the
normalization freedom of $t_{\gi,\mu\nu}$. Eq.~\eqref{eq:gauindep}
expresses the independence of the physical amplitude of the diphoton
decay of the higgs of the choice of the gauge. We
regard~\eqref{eq:gauindep} as a normalization condition of
time-ordered products. We have shown that this condition can be
satisfied in the case at hand and determines uniquely the
indeterminate normalization polynomial of~$t_{\gi,\mu\nu}$ in the
expression~\eqref{eq:gratatio-capitis}.

In summary, our final result for the entirely on-shell EW 
$h \to \ga\ga$ decay reads:
$$
t_{\mu\nu}(k_1,k_2)\bigr|_{\tilde\rho=\rho}
= - \frac{1}{2^3\,(2\pi)^6\,M^2}\, P_{\mu\nu}(k_1,k_2)\, F_1(\rho), 
$$
in agreement with the majority of the literature.

\subsection{On settling the controversy}
\label{ssc:dondedijedigo}

Should one infer that by computing in the ``physical'' unitary gauge
there is no way to entirely settle the controversy that motivates
this work, by removing the remaining ambiguity in determining the
amplitude in question? Not without at least pondering credible
``heavy-higgs'' (or $M \to 0$) and ``light-higgs'' (or $M \to \infty$)
arguments to bolster the case of $F_1(\rho)$ versus $F_1(\rho) - 2$,
that have been made in the literature. 

Now, for the present authors the question is not whether either class
of arguments is compelling enough. Instead, the question is whether
they can be made within the BEG prescriptions, and at the level of
rigour of this paper. The arguments in the first-named class involve
plays with field transformations, power counting rules and the
adiabatic limit that we find hard to countenance in the BEG formalism.

However, those of the second class are persuasive within our purview.
Note that $F(0)$, for both scalar and vector boson charged fields, as
well as for Dirac fermions, must coincide with (the first coefficient
of) the $\beta$-function series associated to \textit{electric charge
renormalization}.%
\footnote{$F_0(0)=-1/3$, which has been calculated in this paper, 
means precisely this.}
It was a fortunate historical fact that a calculation of the effective
Lagrangian for charged Proca particles~\cite{RussianCharge1} was
already available when the first ``exact'' computation of the higgs to
digamma process that we are aware of was performed
\cite{RussianCharge2} -- thus making possible a dependable
``light-higgs'' argument. A computation of the renormalization of the
electric charge of massive vector bosons in the unitary gauge by means
of BEG technology is in principle feasible -- cf. in this
respect~\cite[Sect.~7]{BDF09} and \cite{Duetsch15} -- and expected to
yield the required value $F_1(0)=7$. That would complete the analysis
of this paper, without going beyond the unitary gauge framework.

\section{Conclusion}
\label{sec:die-Eule}

Contrary to custom, we begin this section by declaring what we have
\textit{not} done in the paper. Finite QFT \textit{\`a~la}
Bogoliubov--Epstein--Glaser is mathematically a rigorous method. So,
referring to what is found in the literature -- like that cited in the
Introduction -- we have not employed dimensional regularization,
deemed an ``artifact'' by some. Nor do we borrow Pauli--Villars', nor
cutoff regularizations, for that matter. We did not have to practice
``judicious routings of the external momenta''~\cite{GastmansWuWu2},
nor adopt the ``loop regularization method''~\cite{TresChinosAgudos},
or any of the techniques to handle divergent integrals, resulting from
the blind application of Feynman graph technology on momentum space.
We do not pore over divergent integrals, at all. Each and every one of
the integrals appearing in this paper produces an unambiguous result;
each amplitude is finite.

We expected the BEG procedure to yield a conceptually clear
understanding of the EW $h \to \ga\ga$ decay in the unitary gauge. We
have succeeded in this -- at a price. According to Epstein and Glaser,
the adiabatic limit is to be performed \textit{after} distribution
splitting. Such an off-shell procedure for the $h \to \ga\ga$ decay in
the unitary gauge demands computations more than one order of
magnitude greater than the ones performed in this paper -- compare the
computation of the QED vertex function in~\cite[Chap.~3.8]{Scharf14}
and in~\cite{DuetschKS93}.

We were not disposed to inflict this on ourselves, nor our surviving
readers. Thus we were forced to innovate on the method, generalizing
the splitting dispersion integral to production of massless particles,
and showing that in the present situation the adiabatic limit may be
performed before distribution splitting. Only, then one may have to
add to the result so obtained an a~priori indeterminate polynomial in
the external momenta, of a degree given by the singular order of the
amplitude off-shell. It is precisely the addition of this polynomial
that is missing in references \cite{GastmansWuWu1,GastmansWuWu2} and
\cite{ChristovaI}. We have resolved the ambiguity by recourse to
gauge-fixing independence of the entirely on-shell amplitude.
Alternatively, the ambiguity could be resolved within the unitary
gauge in the BEG scheme, by invoking the low-energy argument.%
\footnote{Variants of the ``light-higgs'' or ``low energy'' argument
besides \cite{RussianCharge2,RussianCharge3} and
\cite{MelnikovVainshtein} are found for instance in
\cite[Ch.~24.8]{GranLev}, in \cite{KSpira} and in \cite{HStoehr}.}
We have not attempted here a rigorous proof of this argument, nor
computed the relevant coefficient of the beta function, leaving the
task for a separate analysis in future work.

\subsection*{Acknowledgements}

We are grateful to E. Alvarez, L. Alvarez-Gaum\'e, M. Herrero, C. P.
Mart\'in, J. C. V\'arilly and T. T. Wu for comments, discussions and
helpful remarks. We particularly thank I.~T.~Todorov for keen help in
the beginning, and his continued and thought-provoking, if contrarian,
interest in this work. As well we thank an anonymous referee for
knowledgeable reporting, definitely contributing to improve the paper.
During the inception and writing of this article, PD received funding
from the National Science Center, Poland, under the grant
UMO-2017/25/N/ST2/01012. He also gratefully acknowledges the
hospitality of the University of Zaragoza. JMG-B received funding from
the European Union's Horizon~2020 research programme under the Marie
Sk{\l}odowska-Curie grant agreement RISE~690575; from Project
FPA2015--65745--P of MINECO/Feder; from CERN; from the COST actions
MP1405 and~CA18108. Hospitality of CERN, IFT-Madrid, ITP-G\"ottingen
and ZiF-Bielefeld is gratefully acknowledged.


\appendix

\section{Notations and prerequisites}
\label{app:spin-one-basics}

Our Minkowski metric is mostly-negative. The Minkowski inner product
of two vectors $x \equiv x^\mu$, $p \equiv p^\nu$ is denoted with
parentheses: $(xp) = x^\mu p_\mu$. When (we hope) it does not cause
confusion, we often denote $p^2 = (pp)$.

We signal the standard formula for time-ordered $2$-point function:
\begin{equation}
\vev{\T \vf(x)\, \chi(x')} 
:= \frac{i}{(2\pi)^4} \int d^4p\,
\frac{e^{-i(p(x - x'))}}{p^2 - M^2 + i0}\, M^{\vf\chi}(p), 
\label{eq:pandemonium} 
\end{equation}
where $M^{\vf\chi}$ is the multiplier appearing in the corresponding
$2$-point function for the fields $\vf,\chi$ with the same mass~$M$.

\smallskip

\textbf{Propagators for a (complex) scalar field.} Clearly, for (say,
complex) \textit{scalar} fields the Feynman propagator
\begin{align}
\Delta^F(x-x') &:= \vev{\T \vf(x)\, \vf^\7(x')}
\label{eq:accidit-in-puncto} 
\\
\shortintertext{fulfils}
\square \Delta^F(x - x') &= -M^2 \Delta^F(x - x') - i\dl(x - x'),
\label{eq:post-factum} 
\end{align}
where $M$ is the mass of the $\vf$-field. Also, with $\th$ denoting
the Heaviside function, the Wightman functions
\begin{gather}
\Dl^+(x - x') := \vev{\vf(x) \vf^\7(x')} = \vev{\vf^\7(x) \vf(x')}
= \frac{1}{(2\pi)^3} \int d^4p\, \th(p^0) \,\dl(p^2 - M^2)
e^{-i(p(x - x'))}
\notag
\\
\word{so that} (\square + M^2)\Dl^+(x) = 0, \qquad 
\Dl^-(x) := -\Dl^+(-x),
\label{eq:nullius-in-verba} 
\end{gather}
are used in our calculations.

\smallskip

\textbf{Massive vector fields.}
A dreibein $e_r(p)$ on Minkowski momentum space, with the properties:
$$
\bigl( e_r(p)\, e_s(p) \bigr) = - \dl_{rs} \word{for} r,s = 1,2,3;
\qquad
\bigl( p\, e_r(p) \bigr) = 0,
$$
describes polarization states for particles with squared mass 
$M^2 = p^2 > 0$ and spin $j = 1$. From the above identities, one
derives the projector formula:
\begin{equation}
\sum_{r=1}^3 e_r^\mu(p) e_r^\nu(p) 
= - g^{\mu\nu} + \frac{p^\mu p^\nu}{M^2}.
\label{eq:Proca-twiner} 
\end{equation}
The set $e$ is regarded as an intertwiner matrix mapping the natural
representation space of the Lorentz group onto the representation
space~$\bC^3$ for spin~$1$ objects. Let $a_r^\7(p)$ and $a_r(p)$ be
respectively the creation and annihilation operators on the boson Fock
space for such particles -- whose $1$-particle subspace is the
corresponding Wigner unirrep space; and $b_r^\7(p)$ and $b_r(p)$ for
their antiparticles.

There is a quantum vector field acting on that space given by the
formula
\begin{equation}
W^\mu(x) := \sum_r \int d\mu(p)\,
\bigl[ e^{i(px)} e_r^\mu(p)\, b^\7_r(p) 
+ e^{-i(px)} e_r^{\mu*}(p)\, a_r(p) \bigr];
\label{eq:Proca-W-field} 
\end{equation}
In \eqref{eq:Proca-W-field} and in other formulas $d\mu(p)$ denotes
the usual invariant measure 
$d^3\pp/2E(p) = d^3\pp/\sqrt{m^2 + |\pp|^2}$ over the mass hyperboloid
$H_M^\pm := \{p\in\bM\,\vert\, p^2 = M^2\wedge\pm p^n > 0\}$. By its
definition, the charged \textit{Proca} field $W$ is divergenceless:
$(\del W) = 0$. Its equations of motion can be variously written as
\begin{equation}
(\square + M^2) W^\mu = (\square + M^2) W^\mu - \del^\mu(\del W) =
\del_\nu G^{\nu\mu}(x) + M^2 W^\mu = 0,
\label{eq:Proca-eqn} 
\end{equation}
where $G^{\mu\nu} := \del^\mu W^\nu - \del^\nu W^\mu$.

The theory of massive vector fields is a gauge theory \cite{Pauli41,
RueggRA04}, its Proca version being a ``unitary gauge'' for it. It has
been analyzed, in terms parallel to Maxwell field theory,
in~\cite{Stora05}; wherein the associated BRST machinery is
``deconstructed'' in terms of Koszul cohomology.

The high-energy limit of $(-g^{\mu\nu} + p^\mu p^\nu/M^2)/(p^2 - M^2)$
apparently signals quadratic divergences and trouble with unitarity of
the scattering matrix: cross-sections would appear to grow without
bound due to the longitudinal momentum states. The difficulty lies
with the closure relation~\eqref{eq:Proca-twiner} of the intertwiners
$e_r$, whose dimension does not allow the standard sufficiency
criterion for renormalizability. This is usually ``cured'' nowadays by
the cohomological extension of the Wigner representation space for
massive spin-$1$ particles into spaces populated by Faddeev-Popov 
ghosts and anti-ghosts and St\"uckelberg fields. 

In this paper we work mainly with the Proca field (i.e., we use the
unitary gauge), where these additional unphysical fields do not
appear; the apparently bad UV-behaviour of the propagators is under
control, as we verify, thanks to amazing cancellations in the
amplitudes.

\smallskip

\textbf{Propagators for the EW theory in the unitary gauge.}
We will make frequent use~of
\begin{align}
\Dl^\al_\bt(x - x') := \vev{\T W^\al(x)\, W_\bt^\7(x')}
= -(g^\al_\bt + \del^\al\del_\bt/M^2) \Delta^F(x - x'),
\label{eq:cum-grano-salis} 
\end{align}
where $M$ is the mass of the $W$-field, and its properties:
$$
\square\Dl^\al_\bt 
= -M^2\Dl^\al_\bt + i(g^\al_\bt + \del^\al\del_\bt/M^2)\dl; \quad
\del_\mu\Dl^\mu_\nu = i\del_\nu\dl/M^2.
$$
The corresponding formulas for the Wightman functions respectively
read:
\begin{gather*}
\Dl^{\al\,+}_\bt(x - x') := \vev{W^\al(x)\, W_\bt^\7(x')}
= \vev{W^{\al\,\7}(x)\, W_\bt(x')} 
= -(g^\al_\bt + \del^\al\del_\bt/M^2) \Dl^+(x - x')
\\
\word{and}
\square\Dl^{\al\,+}_\bt = - M^2 \Dl^{\al\,+}_\bt, \quad
\del_\mu\Dl^{\mu\,+}_\nu = 0.
\end{gather*}
We will invoke also the Maxwell-like fields, where $\# = \7$ or
naught,
$$
F^{\mu\nu} := \del^\mu A^\nu - \del^\nu A^\mu;  \qquad 
G^\#_{\mu\nu} := \del_\mu W^\#_\nu - \del_\nu W^\#_\mu,
$$
and introduce the propagator
\begin{align}
D^{\al\mu}_{\bt\rho}(x - x') 
&:= \vev{\T G^{\al\mu}(x)\, G^\7_{\bt\rho}(x')} 
= \vev{\T G^{\al\mu\,\7}(x)\, G_{\bt\rho}(x')}
\notag \\
&= \vev{\T(\del^\al W^\mu(x) - \del^\mu W^\al(x))
(\del_\bt W^\7_\rho(x') - \del_\rho W^\7_\bt(x'))}
\notag \\
&= -\del^\mu\bigl( \del_\rho\Dl^\al_\bt(x - x') 
- \del_\bt\Dl^\al_\rho(x - x') \bigr)
+ \del^\al\bigl( \del_\rho\Dl^\mu_\bt(x - x')
- \del_\bt\Dl^\mu_\rho(x - x') \bigr)
\notag \\
&= (g^\al_\bt \del^\mu_\rho - g^\al_\rho \del^\mu_\bt
- g^\mu_\bt \del^\al_\rho + g^\mu_\rho \del^\al_\bt) \Delta^F(x - x'),
\quad \del^\al_\rho := \del^\al \del_\rho,
\label{eq:rarior-albo} 
\end{align}
since in
\begin{equation}
\del_\rho \Dl^\al_\bt - \del_\bt \Dl^\al_\rho 
= (-g^\al_\bt \del_\rho + g^\al_\rho \del_\bt) \Delta^F
\label{eq:cancel} 
\end{equation}
the terms with three derivatives \textit{cancel out}, due to the
antisymmetry of~$G_{\mu\nu}^\#$. That fact is relevant in this paper.
Analogously we obtain
\begin{align*}
D^{\al\mu\,+}_{\bt\rho}(x - x') 
&:= \vev{G^{\al\mu}(x)\, G_{\bt\rho}^\7(x')} 
= \vev{G^{\al\mu\,\7}(x)\, G_{\bt\rho}(x')}
\\
&= (g^\al_\bt \del^\mu_\rho - g^\al_\rho \del^\mu_\bt
- g^\mu_\bt \del^\al_\rho + g^\mu_\rho \del^\al_\bt) \Dl^+(x - x'),
\end{align*}
without third-order derivatives. We also note that
$$
\del_\mu D^{\al\mu}_{\bt\rho} 
= (g^\al_\bt \del_\rho - g^\al_\rho \del_\bt)\square \Delta^F 
= (-g^\al_\bt \del_\rho + g^\al_\rho \del_\bt)(M^2 \Delta^F + i\dl),
$$
since third-order derivatives appear only in the form 
$\del\square \Delta^F$, removable with the help 
of~\eqref{eq:post-factum}.

For the $2$-point functions with one $G^\#$ plus one $W^\#$, we use
\eqref{eq:cancel} to get rid of the terms with three derivatives. For
time-ordered ones we obtain
\begin{align}
&\vev{\T W^\mu(x)\, G^\7_{\al\nu}(x')} = \vev{\T W^{\mu\,\7}(x)\,
G_{\al\nu}(x')}
\label{eq:extrema-exquisita} 
\\
&= -(\del_\al\Dl^\mu_\nu(x - x') - \del_\nu\Dl^\mu_\al(x - x'))
= (g^\mu_\nu \del_\al - g^\mu_\al \del_\nu) \Delta^F(x - x'),
\notag 
\\
&\vev{\T G^\7_{\al\nu}(x)\, W^\mu(x')} = \vev{\T G_{\al\nu}(x)\,
W^{\mu\,\7}(x')} = (-g^\mu_\nu \del_\al + g^\mu_\al \del_\nu)
\Delta^F(x - x').
\notag
\end{align}
With the parallel Wightman functions we proceed similarly:
\begin{align}
\vev{W^\mu(x)\, G^\7_{\al\nu}(x')}
&= \vev{W^{\mu\,\7}(x)\, G_{\al\nu}(x')}
 = (g^\mu_\nu \del_\al - g^\mu_\al \del_\nu) \Dl^+(x - x'),
\notag
\\
\vev{G^\7_{\al\nu}(x)\, W^\mu(x')}
&= \vev{G_{\al\nu}(x)\, W^{\mu\,\7}(x')}
 = (-g^\mu_\nu \del_\al + g^\mu_\al \del_\nu) \Dl^+(x - x').
\label{eq:remedia-optima-sunt} 
\end{align}
Comparing with the Feynman gauge, in which the $W^\#$ two-point
functions $\Dl^\al _\bt$ and $\Dl^{\al\,+}_\bt$ are replaced by
$-g^\al_\bt \Delta^F$ \eqref{eq:T(WW)-1}
and $-g^\al_\bt\Dl^+$ \eqref{eq:WW-1}, respectively, we find the
$W^\#G^\#$, $G^\#W^\#$ and $G^\#G^\#$ two-point functions to be the
same, thanks to the cancellations in~\eqref{eq:cancel}.

\smallskip

\textbf{Propagators for the EW theory in the Feynman gauge.}
The Feynman propagator and the Wightman two-point function for the
$W$-field in the Feynman gauge read
\begin{align}
\vev{\T W^\al(x)\, W_\bt^\7(x')}
&= -g^\al_\bt \Delta^F(x - x'),\label{eq:T(WW)-1}
\\
\vev{W^\al(x)\, W_\bt^\7(x')}
= \vev{W^{\al\,\7}(x)\, W_\bt(x')} 
&= -g^\al_\bt \Delta^+(x - x').\label{eq:WW-1}
\end{align}
Besides the $W$-field the computation from Sec.~\ref{ssc:ita} involves
the Stückelberg fields $\varphi^\pm$ and the ghost and anti-ghost
fields $C^\pm,\bar C^\pm$ -- where $\phi^\pm:=\frac{1}{\sqrt{2}}
(\phi_1\pm i\phi_2\bigr)$ for $\phi=\varphi,\,C,\,\bar C$. Below we
list the non-vanishing Feynman propagators and two-point functions for
these fields:
\begin{align}
\vev{\T \varphi^+(x)\, \varphi^-(x')}
&= \Delta^F(x - x'),
\\
\vev{\varphi^+(x)\, \varphi^-(x')}
=\vev{\varphi^-(x)\, \varphi^+(x')}
&=\Delta^+(x - x'),
\\
\vev{\T C^+(x)\, \bar C^-(x')}
=\vev{\T C^-(x)\, \bar C^+(x')}
&= \Delta^F(x - x'), 
\\
\vev{C^+(x)\, \bar C^-(x')}
=-\vev{\bar C^-(x)\, C^+(x')}
&=\Delta^+(x - x'),
\\
\vev{C^-(x)\, \bar C^+(x')}
=-\vev{\bar C^+(x)\, C^-(x')}
&=\Delta^+(x - x').
\end{align}

\section{An interesting distribution}
\label{app:curioser}

In this appendix we study the distribution $f(\rho)$ appearing in the
amplitude of the $h \to \ga\ga$ decay via both scalar QED and
flavourdynamics.

To define $\sqrt{\.} \colon \bC \to \bC$ and $\log \colon \bC \to \bC$
one uses a cut on the negative real axis:
$$
\sqrt{r\,e^{i\vf}} = \sqrt{r}\,e^{i\vf/2}, \quad
\log{re^{i\vf}} = \log{r} + i\vf, 
\quad\text{both with } \vf \in (-\pi,\pi].
$$
The complex function
\begin{equation}
\tilde f : \begin{cases}
\bC \less \bigl( (-\infty,0) \cup (1,\infty) \bigr) \longto \bC
\\
z \longmapsto -\bigl( \log(\sqrt{1 - z} + i\sqrt{z}) \bigr)^2
\end{cases}
\label{eq:tilde-f} 
\end{equation}
is analytic, in view of the two cuts on the real axis. The
distribution $f(\rho)$ is defined by
\begin{equation}
f : [0,\infty) \longto \bC : \rho \longmapsto f(\rho)
:= \tilde f(\rho + i0).
\label{eq:f} 
\end{equation} 
We claim that
\begin{align}
f(\rho) &= (\arcsin\sqrt{\rho}\bigr)^2 
= \biggl[\arctan \frac{\rho}{\sqrt{1 - \rho^2}} \biggr]^2 
\word{for} 0 \leq \rho \leq 1,
\label{eq:f1} 
\\
f(\rho) &= - \frac{1}{4}\, \Biggl[
\log \frac{1 + \sqrt{1 - \rho^{-1}}}{1 - \sqrt{1 - \rho^{-1}}} -i\pi
\Biggr]^2 \word{for} \rho \geq 1,
\label{eq:f2} 
\end{align}
from which one easily obtains the following formula for the imaginary
part:
\begin{equation}
\Im f(\rho) = \th(\rho - 1)\,\frac{\pi}{2} 
\log\frac{\sqrt{\rho} + \sqrt{\rho-1}}{\sqrt{\rho} - \sqrt{\rho-1}} 
= - \th(\rho - 1)\,\frac{\pi}{2}
\log\bigl( 2\rho - 2\sqrt{\rho(\rho - 1)} - 1 \bigr).
\label{eq:consilia-non-sentis} 
\end{equation}
 
The first claim \eqref{eq:f1} follows immediately from the identity
$$
\arcsin\sqrt{\rho} = -i\log\bigl( \sqrt{1-\rho} + i\sqrt{\rho} \bigr)
\word{for} \rho \in [0,1],
$$
which is obvious from $\exp(i\arcsin x) = \sqrt{1 - x^2} + ix, \quad
|x| \leq 1$.

To prove the second claim \eqref{eq:f2}, first note that one has
$\sqrt{1 - (\rho + i0)} = -i\sqrt{\rho - 1}$ for $\rho \geq 1$. Hence,
there holds:
\begin{align*}
&\log\bigl( \sqrt{1 - (\rho + i0)} + i\sqrt{\rho} \bigr) = \log\bigl(
\sqrt{\rho} - \sqrt{\rho - 1} \bigr) + i\pi/2 = \frac{1}{2} \bigl(
\log\bigl( (\sqrt{\rho} - \sqrt{\rho - 1})^2 \bigr) + i\pi \bigr)
\\
&= \frac{1}{2} \biggl( \log \frac{\sqrt{\rho} - \sqrt{\rho - 1}}
{\sqrt{\rho} + \sqrt{\rho - 1}} + i\pi \biggr) - \frac{1}{2} \biggl(
\log \frac{1 + \sqrt{1 - \rho^{-1}}} {1 - \sqrt{1 - \rho^{-1}}} - i\pi
\biggr),
\end{align*}
from which assertion \eqref{eq:f2} follows.

We point out that, for $\rho\in [0,1]$, in the distribution 
$F_0(\rho) = \rho^{-1} \bigl( 1 - \rho^{-1} f(\rho) \bigr)$ in
Eq.~\eqref{eq:magister-dixit} the terms $\sim \rho^{-1}$ cancel. We 
bring in the power series expansion
\begin{align}
\arcsin x &= x + \frac{x^3}{2\.3} + \frac{3\,x^5}{2\.4\.5}
+ \frac{3\.5\,x^7}{2\.4\.6\.7} +\cdots \word{for} |x|\leq 1,
\word{yielding}
\notag
\\
f(\rho) &= (\arcsin \sqrt{\rho}\,)^2 = \rho + \frac{\rho^2}{3} +
\frac{8\rho^3}{45} +\cdots \word{so that} F_0(\rho) = - \frac{1}{3} -
\frac{8}{45}\,\rho +\cdots\,.
\label{eq:arcsin-2} 
\end{align}

\section{Bogoliubov--Epstein--Glaser normalization}
\label{sec:Streu}

Epstein and Glaser \cite{EpsteinGlaser73, EpsteinGlaser76} started
from Bogoliubov's functional $\bS[g]$-matrix
\cite[Sect.~21]{BogoliubovS80}, based on \cite{Bogoliubov55} and on
previous work by St\"uckelberg and Rivier \cite{StueckelbergR50}. That
is an expansion of operator-valued distributions (OVD) on
configuration space, of the form
\begin{equation}
\bS[g] = 1 + \sum_{n=1}^\infty \frac{i^n}{n!} 
\int d^4x_1\cdots d^4x_n\, T_n(x_1,\dots,x_n)\, g(x_1) \cdots g(x_n),
\quad g \in \sS(\bR^4,\bR).
\label{EG-summacumlaude} 
\end{equation}
We have taken $\hbar = 1$. The $g$'s are multiplets of coupling
functions which work as adiabatic cutoffs. The $T_n$,
\textit{symmetric} in their arguments, are identified with
chronological or time-ordered $n$-products. This is Bogoliubov's
version of the summands in the formal Dyson expansion for the
scattering matrix in the interaction picture. One tries to recursively
build the $T_n$ from natural postulates: the ultraviolet problem is
solved in that construction. In the ``adiabatic limit''
$g\!\uparrow\!1$ the functional scattering matrix
\eqref{EG-summacumlaude} is expected to converge to the physical~$\bS$
in suitable senses~\cite{MuyDucho}.

\subsection{The Epstein--Glaser postulates}
\label{ssc:porca-miseria}

\begin{description}
	
\item[Beginning of induction]
The procedure is perturbative, the basic building blocks being finite
sets of quantum free fields on their corresponding Fock spaces.
Precisely, $T_1(x)$ is a Wick polynomial in those and their
derivatives -- a well-defined OVD.%
\footnote{One can think of $T_1$ as an ``interaction Lagrangian''.
However, the Lagrangian mindset is inessential here.}
The coupling constants of the model are included in the~$T_n$, the
expansion being a power series on them. The other postulates shall
enable us to construct the $T_n$ from~$T_1$ by induction on~$n$.

\item[Causality]
This is the key requirement, for which the Epstein--Glaser
manufacturing of TOPs is also called ``causal perturbation theory''.
Let $V_\pm$ and $\ovl{V}_\pm$ respectively denote the open forward and
backward lightcones and their closures. If $g_1,g_2$ are such that
\begin{align*}
&\supp g_2 \cap \bigl( \supp g_1 + \ovl{V}_- \bigr) = \emptyset,
\word{then}  \bS[g_1 + g_2] = \bS[g_2]\,\bS[g_1];
\\
&\text{equivalently,} \quad T_n(x_1,\dots,x_n) = T_r(x_1,\dots,x_r)\,
T_{n-r}(x_{r+1},\dots,x_n) \quad \text{whenever}
\\
&\{x_1,\dots,x_r\} \cap \bigl( \{x_{r+1},\dots,x_n\} + \ovl{V}_-
\bigr) = \emptyset, \text{for all $r$ and $n$ with $1 \leq r \leq n -
1$.}
\end{align*}
This is a powerful postulate, called \textit{causal factorization}. It
means that on large open sets of the $n$-point Minkowski space
$(\bM_4)^{\x n} \equiv \bM^n$ the TOP $T_n$ can be built up from its
lower-order counterparts. In the inductive step of the Epstein--Glaser
method, this requirement uniquely determines $T_n$ on the set of
Schwartz functions $\sS(\bM^n \less \Dl_n)$, in terms of the given
$T_k$ at lower orders $k \leq n - 1$, where $\Dl_n$ is the ``thin''
diagonal $\Dl_n := \set{(x_1,\dots,x_n) : x_1 = x_2 =\cdots = x_n}$.
Perturbative normalization is the \textit{extension} of the
operator-valued distribution $T_n$ from $\sS'(\bM^n \less \Dl_n)$ to
$\sS'(\bM^n)$. The gist of BEG normalization is that in local quantum
field theory this problem finds a solution, the induction process
going through. So there is no need to deal with infinities. The
solution of the extension problem is non-unique: in principle one may
add any OVD which is supported on~$\Dl_n$. All further postulates of
Epstein--Glaser have the purpose of giving guidance for this problem;
hence they may be called ``normalization conditions''.

\item[Causal Wick expansion]
The TOPs are required to satisfy the Wick expansion formula. We
display the latter in terms of the interaction $T_1(x) = \vf^k(x)$,
for $\vf$ a real scalar field:
\begin{align*}
& T_n\bigl(\vf^k(x_1),\dots,\vf^k(x_n)\bigr)
\\
&\enspace
= \sum_{l_1,\dots,l_n=0}^k \binom{k}{l_1} \cdots \binom{k}{l_n}
\vev{T_n(\vf^{k-l_1}(x_1),\dots,\vf^{k-l_n}(x_n))}
\,\vf^{l_1}(x_1) \cdots \vf^{l_n}(x_n)
\end{align*}
with $\vev{\cdots}$ denoting vacuum expectation value. This postulate
reduces the extension problem for the OVD $T_n(\cdots)$ to one of
numerical distributions -- a simpler task.

\item[Poincar\'e Covariance]
Let there be given the standard lifting $U(a,\La)$ to Fock space of
the Poincar\'e unitary irreducible representations (unirreps) on
$1$-particle subspaces. Then
$$
U(a,\La)\,\bS[g]\,U^\7(a,\La) = \bS\bigl[ (a,\La)\.g \bigr],
$$
where $((a,\La)\.g)(x) = g(\La^{-1}(x - a))$.
In particular, translation invariance implies that the coefficients in
the causal Wick expansion depend only on the relative coordinates.
Therefore, the extension problem for the numerical distributions is
step by step simplified to an extension to one point, namely from
$\sS'(\bR^{4(n-1)}\less\{0\})$ to $\sS'(\bR^{4(n-1)})$.

\item[Unitarity \textup{(conservation of probability)}]
\begin{align*}
&\bS[g]\,\bS^\7[g] = \bS^\7[g]\,\bS[g] = 1; \quad \text{here we denote:}
\\
&\bS^{-1}[g] =: 1 + \sum_{n=1}^\infty \frac{(-i)^n}{n!} \int
d^4x_1\cdots d^4x_n\, \ovl{T}_n(x_1,\dots,x_n)\, g(x_1)\cdots g(x_n).
\end{align*}

\item[Divergence degree]
Heuristically, this is the requirement that normalization does not
make the $T$-product ``more singular'' (in the UV-region). This is
expressed in terms of the \textbf{scaling degree} of the coefficients
(i.e., the numerical distributions) in the causal Wick expansion of
the $T$-product: that degree may \textit{not} be increased by the
extension. The standard definitions of the scaling degree $\sd(t)$ and
the singular order $\om(t)$ of a distribution $t \in \sS'(\bR^k)$ or
$t \in \sS'(\bR^k \less \{0\})$ -- see, e.g.,
\cite[Sect.~3.2.2]{Duetsch19} -- are as follows:
\begin{equation}
\sd(t) := \inf\set{r \in \bR : \lim_{\la\downto 0} \la^r\,t(\la x) =
0}, \quad \om(t) := \sd(t) - k,
\label{eq:sd} 
\end{equation}
where $\inf\,\emptyset := \infty$ and $\inf\,\bR := -\infty$. For
instance, for a translation-invariant distribution 
$d(x_1 - x_3, x_2 - x_3) \in \sS'(\bR^8)$ fulfilling $\sd(d) = 8$, 
equivalently $\om(d) = 0$, we say that the amplitude superficially is
``logarithmically divergent''.

\item[Other invariance rules and physical requirements]
Discrete symmetries can be accomodated in the Epstein--Glaser
construction~\cite{Kleopatra}. A Ward identity playing a paramount
role in this paper corresponds to EGI -- see subsections
\ref{ssc:d-scalar} and~\ref{ssc:eius-est-nolle} for this. For
different types of requirements, consult subsections \ref{ssc:ita} and
\ref{ssc:dondedijedigo}.

\end{description}

\subsection{Iterative building of the time-ordered products}
\label{ssc:lost-in-translation}

To assemble the $T_n$ outside of the thin diagonal $\Dl_n$ from the
inductively known $(T_k)_{1\leq k\leq n-1}$ directly by causal
factorization, one would need a partition of unity subordinate to an
open cover of~$\bM^n \less \Dl_n$ -- see \cite{BrunettiF00} and
\cite[Sect.~3.3]{Duetsch19}. This is problematic for practical
computations. For this reason the original Epstein--Glaser
construction \cite{EpsteinGlaser73, Scharf14} is less direct: it
introduces an intermediate $D_n$-distribution having causal support;
and the crucial step is the \textit{splitting} of $D_n$ into its
advanced and retarded parts. This splitting corresponds precisely to
the above-mentioned extension problem, that is, to perturbative
normalization. A decisive advantage of the method is that the problem
is solved in momentum space by a dispersion integral.

To explain the construction, we first express the antichronological
product $\ovl{T}_n$ in terms of the TOPs $(T_k)_{1\leq k\leq n}$. Let
$N = \{x_1,\dots,x_n\}$ and $I \subseteq N$ with $|I| \neq 0$
elements. Define $T_{|I|}(I) = T_{|I|}(x_i : x_i\in I)$. By the
standard inversion of a formal power series with noncommuting terms in
terms of set compositions, we obtain
\begin{equation}
\ovl{T}_{|N|}(N) = \sum_{k=1}^n (-)^{n+k} \sum_{I_1\uplus\cdots\uplus
I_k=N} T_{|I_1|}(I_1) \cdots T_{|I_k|}(I_k),
\label{eq:full-house} 
\end{equation}
where the disjoint union is over nonempty \textit{blocks}~$I_j$.
The terminology of antichronological products is appropriate, since
if $I \cap (J + \ovl{V}_-) = \emptyset$, then 
$\ovl{T}(I \cup J) = \ovl{T}(J)\,\ovl{T}(I)$.

Retarded and advanced products, denoted by $R_n$ and $A_n$
respectively, are the coefficients in the perturbative expansion of
the respective retarded and advanced interacting fields. For them we
follow the convention in the book~\cite{Duetsch19}, identical to that
of~\cite{EpsteinGlaser73} except that $R_n$ and~$A_n$ have an extra
factor~$i^{n-1}$. In general, Bogoliubov's definitions read:
\begin{align}
R_{n+1}(x_1,\dots,x_{n+1}) 
&:= i^n \sum_{I\subset\{1,\dots,n\}} (-1)^{|I|}\,
\ovl{T}_{|I|}(I)\,T_{|I^c|+1}(I^c, x_{n+1}),
\label{eq:quod-non-speratur1} 
\\
A_{n+1}(x_1,\dots,x_{n+1}) 
&:= i^n \sum_{I\subset\{1,\dots,n\}} (-1)^{|I|}\,
T_{|I^c| + 1}(I^c,x_{n+1})\,\ovl{T}_{|I|}(I),
\label{eq:quod-non-speratur2} 
\end{align}
where $I^c := \{1,\dots,n\} \less I$. Epstein and Glaser
\cite{EpsteinGlaser73} prove that $A_{n+1},\,R_{n+1}$ have advanced or
retarded support, respectively:
\begin{align*}
\supp A_{n+1} &\subseteq 
\set{x \in \bM^{n+1} : x_j - x_{n+1} \in \ovl{V}_+ \ \forall j};
\\
\supp R_{n+1} &\subseteq 
\set{x \in \bM^{n+1} : x_j - x_{n+1} \in \ovl{V}_- \ \forall j}.
\end{align*}
In the induction step $n \to n+1$ neither the $T_{n+1}$ nor the
$R_{n+1}$ nor the $A_{n+1}$ are known. But by the induction hypothesis
the \textit{difference} $D_{n+1}$, defined by $D_{n+1} := A_{n+1} -
R_{n+1}$, only \textit{depends on known quantities}. For instance, in
$D_3$ the unknown~$T_3$ has dropped out -- and $\ovl{T}_1$,
$\ovl{T}_2$ are uniquely given in terms of $T_1$ and~$T_2$. It follows
that $D_{n+1}$ has causal support:
$$
\supp D_{n+1}  \subseteq 
\set{x \in \bM^{n+1} : x_j - x_{n+1} \in \ovl{V}_+ \ \forall j}
\cup \set{x \in \bM^{n+1} : x_j - x_{n+1} \in \ovl{V}_- \ \forall j}.
$$
If one finds a way to extract the advanced part $A_{n+1}$
of~$D_{n+1}$, that is, to \textit{split} the OVD $D_{n+1}$ into
$A_{n+1}$ and $-R_{n+1}$ in such a way that the latter two satisfy the
just given support properties, then one can construct a candidate for
$T_{n+1}$.%
\footnote{That sometimes needs to be symmetrized, by adding a 
suitable OVD supported on $\Dl_{n+1}$.}

For the sake of normalization conditions, at this stage we may add to
$T_{n+1}$ any OVD supported on $\Dl_{n+1}$ which is symmetric in
$x_1,\dots,x_{n+1}$. The $D_{n+1}$ fulfils all the normalization
conditions, in particular the `Causal Wick expansion' and `Translation
invariance', because of the validity of those for the inductively
given $(T_k)_{1\leq k\leq n}$. Therefore, the splitting problem for
$D_{n+1}$ translates into a consonant problem for the coefficients
$d(x_1 - x_{n+1},\dots, x_n - x_{n+1}) \in \sS'(\bR^{4n},\bC)$ in the
Wick expansion of $D_{n+1}$, yielding $a,r(x_1 - x_{n+1},\dots, x_n -
x_{n+1})\in \sS'(\bR^{4n},\bC)$, which are the coefficients in the
Wick expansion of $A_{n+1}$ and $R_{n+1}$, respectively.

\textit{In fine}, by the induction process, one specifies the
ambiguity in the vacuum expectation value of each~$T_{n+1}$ by adding
to it a \textit{contact term}, that is,
\begin{equation}
t(x_1 - x_{n+1},\dots, x_n - x_{n+1}) 
+ \sum_{|a|\leq\om} c_a\,\del^a\dl(x_1 - x_{n+1},\dots,x_n- x_{n+1}),
\label{eq:plus-actum} 
\end{equation}
where $\om$ is the singular order of the pertinent $d(x_1 -
x_{n+1},\dots)$ and the coefficients $c_a \in \bC$ depending on the
multi-index $a$ \textbf{are arbitrary}, up to restrictions coming from
the `Poincar\'e covariance' and `Other invariance rules' requirements.

\subsection{Dispersion integrals from splitting in BEG normalization:
the central solution}
\label{sec:ipso-facto}

For simplicity, here we restrict ourselves to the case of two
four-variables, relevant for this paper.
For the Fourier transform of $f \in \sS\bigl(\bR^8\bigr)$ we employ
the following convention:
\begin{align}
f(y_1,y_2) = (2\pi)^{-4} \int dk_1\,dk_2\, e^{-i(k_1y_1 + k_2y_2)}\,
\hat f(k_1,k_2).
\label{eq:def_Fourier} 
\end{align}
Let $\Ga_\pm := \ovl V_\pm \x \ovl V_\pm$ henceforth. Given a ``causal
distribution'', that is, $d \in \sS'(\bR^8)$ with
\begin{equation}
\supp d \subseteq \Ga_+ \cup \Ga_- \word{and} \sd(d) < \infty,
\label{eq:caus-d} 
\end{equation} 
by a \textit{splitting solution} of~$d$ we mean a distribution $a \in
\sS'(\bR^8)$ with
\begin{equation}
(a - d)\bigr|_{\sS(\bR^8 \less \Ga_-)} = 0, \quad
\supp a \subseteq \Ga_+ \word{and} \sd(a) \leq \sd(d).
\label{eq:split} 
\end{equation}

In what follows we assume that the Fourier transform $\hat d$ of the
causal $d$-distribution we wish to split vanishes in an open ball $\sR
\subset \bR^8$ centered at $k = 0$. This holds if all propagators
contributing to $d$ are massive, as it is the case in this paper --
see \cite[Sect.~5.2]{EpsteinGlaser73}. Also in \cite{EpsteinGlaser73}
it is shown for any splitting solution~$a$ that $\hat d|_{\sR}=0$ entails
analyticity of $\hat a(k)$ on $\sR$. In this case there exists a
distinguished splitting solution, the so-called
\textit{central~solution} $a^c$, characterized by the conditions
\begin{equation}
\del^a\hat a^c(0) = 0, \word{for all} |a| \leq \om(d).
\label{eq:central-sol} 
\end{equation}
As indicated in Eq.~\eqref{eq:plus-actum}, for $\sd(d) \geq 8$ --
i.e., for $\om(d) \geq 0$, as defined in Eq.~\eqref{eq:sd} -- the
splitting solution of~$d$ is not uniquely determined. Any two
solutions $a_1$ and~$a_2$ differ by
$$
a_1(y) - a_2(y) = \sum_{|a|=0}^{\om(d)} C_a\,\del^a\dl(y)
\word{or equivalently,} \hat a_1(k) - \hat a_2(k) 
= \frac{1}{(2\pi)^4} \sum_{|a|=0}^{\om(d)} C_a\,(-ik)^a,
$$
with arbitrary constants $C_a \in \bC$.

Essential for dealing with our situation is that the central solution
of the splitting problem in momentum space \textit{can be computed by
a dispersion integral}. Now we sketch the derivation of a few
versions of this distinguished splitting integral.%
\footnote{For further detail we refer to \cite[Sect.~3.2]{Scharf14},
which relies on \cite[Sect.~6.5]{EpsteinGlaser73}.}
The naive way to extract the advanced part $a$ of $d$ is to multiply
the latter by a $\th$-function:
$$
a_\naive(y_1,y_2) := d(y_1,y_2)\,\chi(y_1,y_2) \word{with}
\chi(y_1,y_2) := \th\bigl( (y_1v_1) + (y_2v_2) \bigr),
$$
where $v := (v_1,v_2) \in V_+ \x V_+$ is arbitrary. But for $\sd(d)
\geq 8$ the pointwise product $d\chi$ exists only as an element of
$\sS'(\bR^8 \less \{0\})$. Therefore, the splitting problem \textit{is
an extension problem}: we have to extend $d\chi \in \sS'(\bR^8 \less
\{0\})$ to an $a \in \sS'(\bR^8)$ such that $\sd(a) = \sd(d\chi) =
\sd(d)$.%
\footnote{A priori, it might happen that $\sd(d\chi) < \sd(d)$; but in
the applications to Epstein--Glaser normalization known to us one
always finds $\sd(d\chi) = \sd(d)$. Hence we assume the latter
relation to hold true.}

The problem is studied in its particulars in
\cite[Sect.~3.2.2]{Duetsch19}. Given $a$ with singular order~$\om$,
there exists an obvious extension $a^\om$, belonging in the dual space
$\sS_\om'(\bR^8)$ of
$$
\sS_\om(\bR^8) := \set{f \in \sS(\bR^8)
: \del^b f(0) = 0~~\text{for all}~~|b| \leq \om},
$$
uniquely determined by the requirement that $\sd(a^\om) = \sd(d)$.
Next, a projection is introduced:
\begin{equation}
W_\om : \sS(\bR^8) \longto \sS_\om(\bR^8); \qquad W_\om f(y) := f(y) -
w(y) \sum_{|b|=0}^\om \frac{y^b}{b!}\, \del^b f(0),
\label{eq:shinto} 
\end{equation}
where the suitably decaying function $w$ must fulfil $w(0) = 1$ and
$\del^b w(0) = 0$ for $1 \leq |b| \leq \om$. One verifies that a
solution $a_w$ (depending on the choice of the function~$w$) of the
splitting problem \eqref{eq:split} is obtained by setting
\begin{equation} 
\braket{a_w}{f} := \braket{a^\om}{W_\om f}.
\label{eq:W-ext} 
\end{equation}
The $a^\om$ involved here is~$d\chi$ with enlarged domain. If
furthermore assumption $d|_{\sR}=0$ is satisfied, the infrared behaviour
of~$d(y)$ is harmless. Hence, one may simply choose $w(y) = 1$ for all
$y \in \bR^4$. Then the correspondent splitting solution $a_{w=1}$ is
actually the central solution $a^c$~\eqref{eq:central-sol}.
Substituting $d\chi$ for $a^\om$ and further using the convolution
formula
$$
\hat f \star \hat\chi(k) = (2\pi)^4 \frac{i}{2\pi} \int_\bR
\frac{dt}{t + i0}\, \hat f(k - tv)
$$
and $\widehat{fg} = (2\pi)^{-4} \hat{f} \star \hat{g}$, we see that
\begin{equation}
\hat a^c(k) = \frac{i}{2\pi} \int_\bR \frac{dt}{t + i0}\, \biggl[
\hat d(k - tv) - \sum_{|b|=0}^\om \frac{k^b}{b!}\, 
\del^b\hat d(-tv) \biggr].
\label{eq:ac} 
\end{equation}
This splitting integral does not depend on the choice of $v \in V_+\x
V_+$. Moreover, for $k \in V_\eta \x V_\eta$, where $\eta \in
\{+,-\}$, we may choose $v := \eta k$ -- that $v$ vary with~$k$ is
admissible. With some extra work \cite[Prop.~3.4]{Scharf14}, this
formula is then simplified into a \textit{convergent} dispersion
integral:
\begin{equation}
\hat a^c(k) = \frac{i\eta}{2\pi} \int_\bR dt\,
\frac{\hat d(tk)}{(t - \eta i0)^{\max\{\om+1,\,0\}}(1 - t + i\eta 0)}
\word{for} k \in V_\eta \x V_\eta \,.
\label{eq:hoist-with-retard} 
\end{equation}
In the applications treated in this paper, $d(tk)$ is of the form
\begin{equation}
\hat d(tk) = \eta\,\sgn(t)\,\th(t^2 - t_{\min}^2)\, f\bigl( t^2 k_1^2,
t^2 k_2^2, t^2 (k_1 + k_2)^2 \bigr) \word{for} k\in V_\eta \x V_\eta
\,,
\label{eq:d-form} 
\end{equation}
for some $f \in \sS'(\bR^3)$, where $t_{\min} > 0$ depends on the
squares of the momenta. So finally, introducing the new integration
variable $u := t^2$, the integral \eqref{eq:hoist-with-retard} goes
over into
\begin{equation}
\hat a^c(k) = \frac{i}{2\pi} \int_{t_{\min}^2}^\infty du\,
\frac{f(uk_1^2, uk_2^2, u(k_1 + k_2)^2)}
{u^{\max\{\piso{\om/2}+1,\,0\}}(1 - u + i\eta 0)}
\word{for} k \in V_\eta \x V_\eta \,,
\label{eq:pro-reo} 
\end{equation}
where $\piso{\cdot}$ denotes the integer part.


\end{document}